\documentclass[a4paper, 12 pt]{article}
\usepackage[margin = 0.75 in, includefoot]{geometry} 
\usepackage{setspace} 
\usepackage{amsmath, amssymb, amsthm, mathrsfs} 
\usepackage[title]{appendix} 
\usepackage[affil-it]{authblk} 
\usepackage{blkarray, bigstrut} 
\usepackage{chngcntr} 
\usepackage{enumerate} 
\usepackage{graphicx} 
\usepackage{multicol} 
\usepackage{quotes} 
\usepackage{xcolor} 
\usepackage{hyperref} 

\newcommand\topstrut[1][1.2ex]{\setlength\bigstrutjot{#1}{\bigstrut[t]}} 
\newcommand\botstrut[1][0.9ex]{\setlength\bigstrutjot{#1}{\bigstrut[b]}} 
\newcounter{nameOfYourChoice} 
\newtheorem{theorem}{Theorem} 
\newtheorem{lemma}[theorem]{Lemma} 
\newtheorem{proposition}[theorem]{Proposition} 
\theoremstyle{remark}
\newtheorem{remark}{Remark} 
\newcounter{example} 
\newenvironment{example}[1][]{\refstepcounter{example}\par
	\noindent \textbf{Example~\theexample. #1}}{\bigskip}
\allowdisplaybreaks 

\title{\textbf{Reaction Network Analysis \\ of Metabolic Insulin Signaling}}

\author[1]{\textbf{Patrick Vincent N. Lubenia}}
\author[1,2,3]{\textbf{Eduardo R. Mendoza}}
\author[1,2,4\thanks{Corresponding author. \\ \indent \hspace{0.01 cm} Emails: angelyn.lao@dlsu.edu.ph, pnlubenia@upd.edu.ph, eduardo.mendoza@dlsu.edu.ph}]{\textbf{Angelyn R. Lao}}

\affil[1]{Systems and Computational Biology Research Unit, Center for Natural Sciences and Environmental Research, 2401 Taft Avenue, Manila, 0922, Metro Manila, Philippines}
\affil[2]{Department of Mathematics and Statistics, De La Salle University, 2401 Taft Avenue, Manila, 0922, Metro Manila, Philippines}
\affil[3]{Max Planck Institute of Biochemistry, Am Klopferspitz 18, 82152, Martinsried near Munich, Germany}
\affil[4]{Center for Complexity and Emerging Technologies, 2401 Taft Avenue, Manila, 0922, Metro Manila, Philippines}

\date{}

\begin{document}

\maketitle

\begin{abstract}
	Absolute concentration robustness (ACR) and concordance are novel concepts in the theory of robustness and stability within Chemical Reaction Network Theory. In this paper, we have extended Shinar and Feinberg's reaction network analysis approach to the insulin signaling system based on recent advances in decomposing reaction networks. We have shown that the network with 20 species, 35 complexes, and 35 reactions is concordant, implying at most one positive equilibrium in each of its stoichiometric compatibility class. We have obtained the system's finest independent decomposition consisting of 10 subnetworks, a coarsening of which reveals three subnetworks which are not only functionally but also structurally important. Utilizing the network's deficiency-oriented coarsening, we have developed a method to determine positive equilibria for the entire network. Our analysis has also shown that the system has ACR in 8 species all coming from a deficiency zero subnetwork. Interestingly, we have shown that, for a set of rate constants, the insulin-regulated glucose transporter GLUT4 (important in glucose energy metabolism), has stable ACR.

\bigskip

\noindent
\textbf{Keywords:} independent decomposition, metabolic insulin signaling, positive equilibria, reaction network, subnetwork
\end{abstract}

\section{Introduction}

Between 2010 and 2015, Shinar and Feinberg published a series of papers based on the novel concepts of absolute concentration robustness (ACR) and concordance, which one may view as the beginnings of a theory of robustness within Chemical Reaction Network Theory (CRNT) \cite{KSHF2015, SHFE2010, SHFE2011, SHFE2012, SHFE2013}. ACR characterizes the invariance of the concentrations of a species at all positive equilibria of a kinetic system and, from experimental observations in \textit{Escherichia coli} subsystems, the authors extracted sufficient (mathematical) conditions for the property. They related ACR to bifunctionality of enzymes and viewed the condition as a ``structural source'' \cite{SHFE2010} or ``design principle'' \cite{SHFE2011} of robustness. Concordance, on the other hand, is seen to indicate ``architectures that, by their very nature, enforce duller, more restrictive behavior despite what might be great intricacy in the interplay of many species, even independent of values that kinetic parameters might take'' \cite{SHFE2012}. Shinar and Feinberg provided small models of biological systems such as the osmotic pressure response system EnvZ-OmpR of \textit{Escherichia coli} for ACR and calcium dynamics of olfactory cilia for concordance. The goal of this paper is to explore the extension of their reaction network analysis approach to larger models of biochemical systems based on recent advances in decomposing reaction networks \cite{FOMF2021, FOME2021, HEDC2021} using the example of a widely used model of the insulin signaling system \cite{SEDA2002}.

The insulin signaling system is an important metabolic system that, upon binding of insulin to its receptor at the cell surface, initiates the uptake of glucose into the cell. The analysis of this process in muscle cells, hepatocytes, cells of adipose tissue, and (most recently) neurons is crucial for understanding the underlying mechanisms of insulin resistance. This reduced ability of cells to use available insulin for energy metabolism is viewed as a common factor in diseases such as obesity, type 2 diabetes, metabolic syndrome, and cancer. More recently, brain insulin resistance has received increased attention from neuroscientists in connection with mild cognitive impairment and Alzheimer's Disease (AD) \cite{ARNO2018, CRAF2017}. Our studies of the metabolic aspects of AD, in fact, provided the motivation for this analysis of a mathematical model of the insulin signaling system \cite{VTML2022}.

The complexity of the insulin signaling system, both in terms of the number of molecular components involved as well as the intricate combination of positive and negative feedback loops, clearly warrants the application of mathematical modeling and computational tools. A natural starting point for our studies of such models is the seminal work of Sedaghat et al \cite{SEDA2002}. This widely-cited work has been applied in various contexts and the authors conveniently provided WinPP source files for the system. We utilized the Hars-T\'{o}th criterion presented in \cite{CBHB2009} to derive a chemical reaction network underlying the system of ordinary differential equations (ODEs) of the Sedaghat et al model, leading to its realization as a mass action system with 20 species, 35 complexes, and 35 reactions.

The first main result of our analysis of the Sedaghat et al system is that the underlying network is concordant. Concordance is a property abstracted from the continuous flow stirred tank reactor model widely used in chemical engineering and its occurrence in complex biological systems is not straightforward. To date, only two smaller systems---the previously-mentioned calcium signaling and the Wnt signaling \cite{FEIN2019}---have been shown to have the property. Concordance has numerous structural and kinetic consequences including monostationarity for the insulin signaling system (see Section \ref{sec:basic}). A detailed discussion of concordance properties can be found in Chapter 10 of Feinberg's recent book \cite{FEIN2019}.

The remaining main results derive from the systematic use of decomposition theory.

Decomposition theory was initiated by Feinberg in his 1987 review \cite{FEIN1987} where he also introduced the important concept of an independent decomposition, a decomposition wherein the stoichiometric subspace of the whole network is the direct sum of the stoichiometric subspaces of the subnetworks. He highlighted its importance by showing that in an independent decomposition the set of positive equilibria of the whole network is the intersection of the equilibria sets of the subnetworks. Hence, if a species has ACR in a subnetwork of an independent decomposition, it also has ACR in the whole network since the latter's equilibria set is contained in that of the subnetwork. Recent results of Hernandez and De la Cruz \cite{HEDC2021} provided a criterion and procedure for determining the existence of a nontrivial independent decomposition (the trivial independent decomposition is the network itself).

The second main result is the existence of nontrivial independent decompositions of the network. This property allows insights into structural relationships between positive equilibria of the whole network and its subnetworks. Since the algorithm from Hernandez and De la Cruz \cite{HEDC2021} determines the finest independent decomposition and independence is invariant under decomposition coarsening, all independent decompositions can be generated from the result. We developed a Matlab code for the said algorithm and applied it to the insulin signaling system.

Third, it is shown that the networks of three functional modules---the insulin receptor binding and recycling subsystem, the postreceptor signaling subsystem, and the GLUT4 translocation subsystem---discussed by Sedaghat et al also form an independent decomposition. These subsystems, hence, not only are functionally but also structurally significant, i.e., positive equilibria of the whole system come from equilibria of these systems. Further details can be found in Section \ref{subsec:finest}.

Fourth is the existence of a large weakly reversible, deficiency zero subnetwork constituting a well-understood part of the system and which is also the source of all 8 ACR species.

For the ACR analysis, we implemented the algorithm of Fontanil et al \cite{FOMF2021} in Matlab and discovered Shinar-Feinberg reaction pairs in appropriate low-deficiency subnetworks of coarsenings of the finest independent decomposition. We found that the system has ACR in 8 (out of 20) species. This restricts the variability of the positive equilibria and suggests that this ``structural source of robustness'' may be an important factor in the system's overall robustness, a property which, according to a previous study by Dexter et al \cite{DEXD2015}, is not common in biochemical systems. Overall, however, there is still a high variation in equilibria composition due to the lack of ACR species in the deficiency 7 subnetwork of the network's deficiency-oriented decomposition. To our knowledge, this is the first large kinetic system for which an ACR assessment has been documented.

The last main result is the discovery of ACR of the essential glucose transporter GLUT4 which, coupled with adequate glucose supply, enables reliable cellular energy production. Further details are discussed in Section \ref{sec:acrDef}.

The paper is organized as follows: The construction of the metabolic insulin signaling reaction network, as well as its basic properties are described in Section \ref{sec:model}. Section \ref{sec:indepDecomp} discusses the finest nontrivial independent decomposition of the insulin signaling system and its deficiency-oriented coarsening. In Section \ref{sec:acr}, a brief review of relevant ACR results is followed by ACR analyses based on decompositions of the network. A short summary and the future direction of our research conclude the paper in Section \ref{sec:conclusion}. Basic concepts and notations related to CRNT, as well as some detailed computations relevant to the paper, are provided in the Appendix.

\section{Model Realization as a Mass Action System}
\label{sec:model}

In this section, we present the ODE system of Sedaghat et al \cite{SEDA2002} and its mass action system realization. We then analyze various properties of the underlying chemical reaction network (CRN), including its positive dependency, $t$-minimality, nonconservativity, and concordance.

\subsection{Insulin Signaling Model}
\label{subsec:ode}

We consider the system of ODEs in Sedaghat et al \cite{SEDA2002} which has only the state variables $X_2$ to $X_{21}$. We streamlined the notation to clearly identify the 35 reactions involved in the  metabolic insulin signaling network (see Appendix \ref{app:var} for the description and units of the variables):
\begin{align*}
	& \dot{X}_2 = k_2 X_3 + k_6 X_5 - k_1 X_2 + k_8 X_6 - k_7 X_2 \\
	& \dot{X}_3 = k_1 X_2 - k_2 X_3 - k_5 X_3 \\
	& \dot{X}_4 = k_3 X_5 - k_4 X_4 + k_{10} X_7 - k_9 X_4 \\
	& \dot{X}_5 = k_5 X_3 + k_4 X_4 - k_3 X_5 - k_6 X_5 + k_{12} X_8 - k_{11} X_5 \\
	& \dot{X}_6 = k_{13} - k_{14} X_6 + k_{15} X_7 + k_{16} X_8 + k_7 X_2 - k_8 X_6 \\
	& \dot{X}_7 = k_9 X_4 - k_{10} X_7 - k_{15} X_7 \\
	& \dot{X}_8 = k_{11} X_5 - k_{12} X_8 - k_{16} X_8 \\
	& \dot{X}_9 = k_{19} X_{10} - k_{17} X_9 X_4 - k_{18} X_9 X_5 \\
	& \dot{X}_{10} = k_{17} X_9 X_4 + k_{18} X_9 X_5 + k_{21} X_{12} - k_{19} X_{10} - k_{20} X_{11} X_{10} \\
	& \dot{X}_{11} = k_{21} X_{12} - k_{20} X_{10} X_{11} \\
	& \dot{X}_{12} = k_{20} X_{10} X_{11} - k_{21} X_{12} \\
	& \dot{X}_{13} = k_{22} X_{12} X_{14} + k_{24} X_{15} - k_{23} X_{13} - k_{25} X_{13} \\
	& \dot{X}_{14} = k_{23} X_{13} - k_{22} X_{12} X_{14} \\
	& \dot{X}_{15} = k_{25} X_{13} - k_{24} X_{15} \\
	& \dot{X}_{16} = k_{27} X_{17} - k_{26} X_{13} X_{16} \\
	& \dot{X}_{17} = k_{26} X_{13} X_{16} - k_{27} X_{17} \\
	& \dot{X}_{18} = k_{29} X_{19} - k_{28} X_{13} X_{18} \\
	& \dot{X}_{19} = k_{28} X_{13} X_{18} - k_{29} X_{19} \\
	& \dot{X}_{20} = k_{31} X_{21} - k_{30} X_{20} - k_{32} X_{17} X_{20} - k_{33} X_{19} X_{20} + k_{34} - k_{35} X_{20} \\
	& \dot{X}_{21} = k_{30} X_{20} + k_{32} X_{17} X_{20} + k_{33} X_{19} X_{20}  - k_{31} X_{21}
\end{align*}

Sedaghat et al used a biochemical map to derive the ODE system and an incomplete extract of the map is shown in their paper. For better visual orientation of the reader, we combined information from the paper's text and the ODE system to reconstruct the complete biochemical map shown in Figure \ref{fig1}.

\begin{figure}[ht]
    \centering
    \includegraphics[width = 0.75\textwidth]{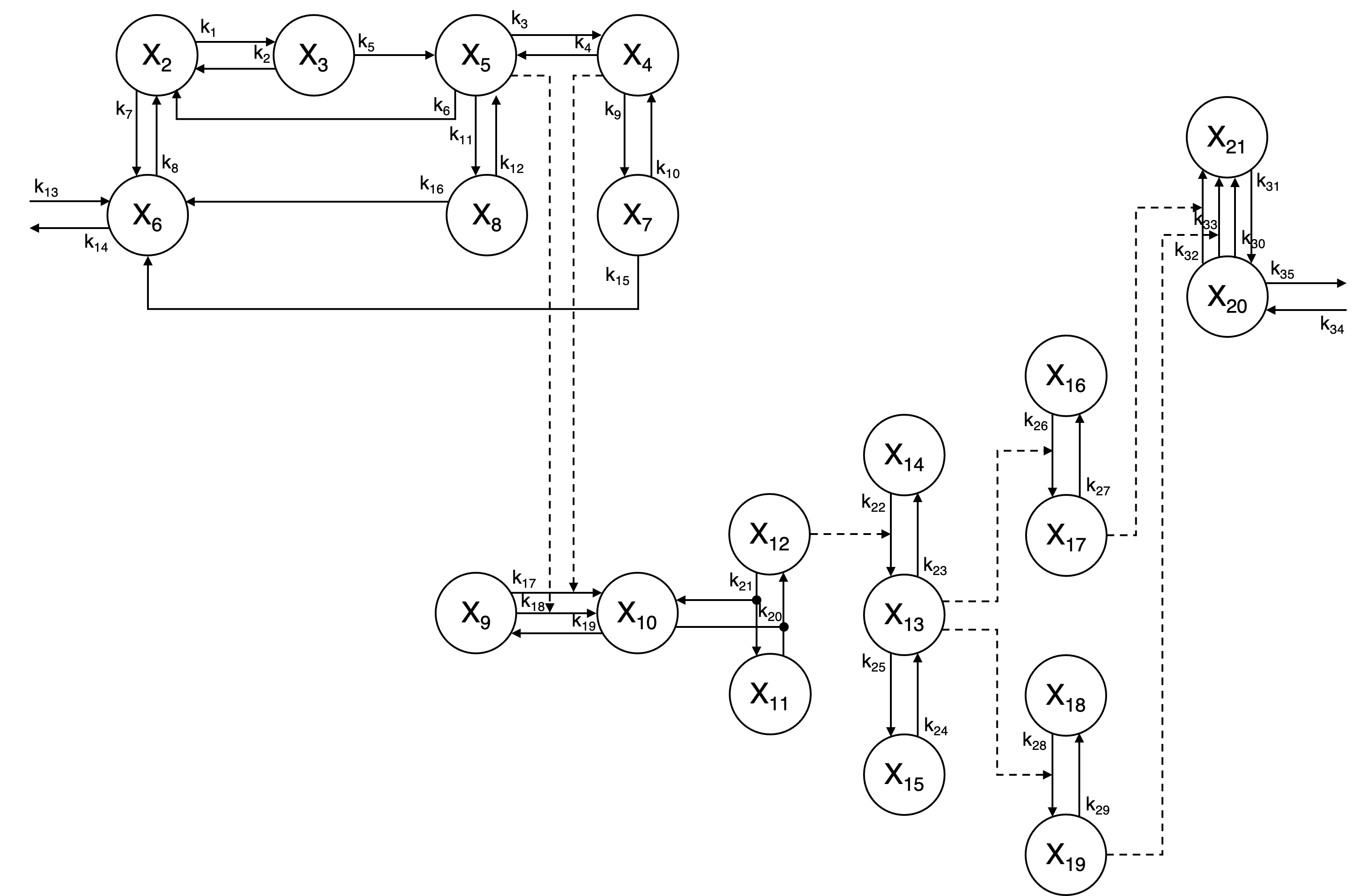}
    \caption{Biochemical map of the insulin signaling system where $X_2, \ldots, X_{21}$ are the species of the network (see Appendix \ref{app:var} for the description and units of the variables); $k_1, \ldots, k_{35}$ are the rate constants of the reactions; and solid lines represent mass transfer reactions while broken lines represent regulatory reactions}
    \label{fig1}
\end{figure}

\subsection{Application of the Hars-T\'{o}th Criterion for Mass Action System Realization}

Theorem 4 of Chellaboina et al \cite{CBHB2009}, called the Hars-T\'{o}th criterion, guarantees a mass action system realization such that the ODE system with species $X_1, \ldots, X_m$ and $r$ reactions is given by $\dot{X} = f(X) = (B - A)^\top (k \circ X^A)$ where $\circ$ represents componentwise multiplication, $A = [a_{ij}]$, $B = [b_{ij}]$, $k = [k_1, \ldots, k_r]^\top$, $X = [X_1, \ldots, X_m]^\top$, and $X^A$ is the element of $\mathbb{R}^m$ with $i$th component $X_1^{a_{i1}} \cdots X_m^{a_{im}}$ for $i = 1, \ldots, r$ and $j = 1, \ldots, m$.

In the insulin signaling model $\dot{X} = f(X)$ in Section \ref{subsec:ode}, it is easy to check that the function $f_j (X_2, \ldots, X_{j-1}, 0, X_{j+1}, \ldots, X_{21})$ is a multivariate polynomial with nonnegative coefficients for each $j = 2, \ldots, 21$. For instance, for $j = 2$ and $j = 3$, we have
\begin{align*}
	& f_2 (0, X_3, X_4 \ldots, X_{21}) = k_2 X_3 + k_6 X_5 + k_8 X_6 \\
	& f_3 (X_2, 0, X_4, \ldots, X_{21}) = k_1 X_2.
\end{align*}
Similar computations for $j = 4, \ldots, 21$ yield the same conclusion. Therefore, by the Hars-T\'{o}th criterion, there is a reaction network of the form $AX \xrightarrow{k} BX$ such that $f(X) = (B - A)^\top (k \circ X^A)$ where $A$ and $B$ have nonnegative integer entries (see Appendix \ref{app:crn} for the detailed computation). The CRN corresponding to the insulin signaling model is:
\begin{multicols}{2}
\noindent
\begin{align*}
	& R_1: X_2 \rightarrow X_3 \\
	& R_2: X_3 \rightarrow X_2 \\
	& R_3: X_5 \rightarrow X_4 \\
	& R_4: X_4 \rightarrow X_5 \\
	& R_5: X_3 \rightarrow X_5 \\
	& R_6: X_5 \rightarrow X_2 \\
	& R_7: X_2 \rightarrow X_6 \\
	& R_8: X_6 \rightarrow X_2 \\
	& R_9: X_4 \rightarrow X_7 \\
	& R_{10}: X_7 \rightarrow X_4 \\
	& R_{11}: X_5 \rightarrow X_8 \\
	& R_{12}: X_8 \rightarrow X_5 \\
	& R_{13}: 0 \rightarrow X_6 \\
	& R_{14}: X_6 \rightarrow 0 \\
	& R_{15}: X_7 \rightarrow X_6 \\
	& R_{16}: X_8 \rightarrow X_6 \\
	& R_{17}: X_9 + X_4 \rightarrow X_{10} + X_4 \\
	& R_{18}: X_9 + X_5 \rightarrow X_{10} + X_5 \\
	& R_{19}: X_{10} \rightarrow X_9 \\
	& R_{20}: X_{10} + X_{11} \rightarrow X_{12} \\
	& R_{21}: X_{12} \rightarrow X_{10} + X_{11} \\
	& R_{22}: X_{14} + X_{12} \rightarrow X_{13} + X_{12} \\
	& R_{23}: X_{13} \rightarrow X_{14} \\
	& R_{24}: X_{15} \rightarrow X_{13} \\
	& R_{25}: X_{13} \rightarrow X_{15} \\
	& R_{26}: X_{16} + X_{13} \rightarrow X_{17} + X_{13} \\
	& R_{27}: X_{17} \rightarrow X_{16} \\
	& R_{28}: X_{18} + X_{13} \rightarrow X_{19} + X_{13} \\
	& R_{29}: X_{19} \rightarrow X_{18} \\
	& R_{30}: X_{20} \rightarrow X_{21} \\
	& R_{31}: X_{21} \rightarrow X_{20} \\
	& R_{32}: X_{20} + X_{17} \rightarrow X_{21} + X_{17} \\
	& R_{33}: X_{20} + X_{19} \rightarrow X_{21} + X_{19} \\
	& R_{34}: 0 \rightarrow X_{20} \\
	& R_{35}: X_{20} \rightarrow 0.
\end{align*}
\end{multicols}

\subsection{CRNT Analysis of the Mass Action System: Key Basic Properties}
\label{sec:basic}

From this point onward, we denote the CRN constructed in the previous section for the metabolic insulin signaling network as $\mathscr{N} = (\mathscr{S}, \mathscr{C}, \mathscr{R})$ with mass action kinetics $K$, stoichiometric subspace $S$, set of species $\mathscr{S} = \{ X_2, \ldots, X_{21} \}$, set of complexes $\mathscr{C}$, and set of reactions $\mathscr{R} = \{ R_1, \ldots, R_{35} \}$.

Table \ref{tab:INSMS} presents the network numbers of $\mathscr{N}$. These numbers provide a cursory analysis of the structural and dynamical properties of the CRN. Properties inferred by Table \ref{tab:INSMS} include non-weak reversibility (since $s \ell = 24 \neq 13 = \ell$), $t$-minimality (since $t = 13 = \ell$), branching (since $r = 35 > 24 = n_r$), and non-point terminality and non-cycle terminality (since $t \neq n - n_r = 11$ and $n - n_r \neq 0$, respectively).

\begin{table}[ht]
    \begin{center}
	    \caption{Network numbers of $\mathscr{N}$}
	    \label{tab:INSMS}
	    \begin{tabular}{@{}llc@{}}
	        \hline
			Network numbers & & \\
			\hline
			Species & $m$ & 20 \\
			Complexes & $n$ & 35 \\
			Reactant complexes & $n_r$ & 24 \\
			Reversible reactions & $r_\text{rev}$ & 10 \\
			Irreversible reactions & $r_\text{irrev}$ & 15 \\
			Reactions & $r$ & 35 \\
			Linkage classes & $\ell$ & 13 \\
			Strong linkage classes & $s \ell$ & 24 \\
			Terminal strong linkage classes & $t$ & 13 \\
			Rank & $s$ & 15 \\
			Reactant rank & $q$ & 20 \\
			Deficiency & $\delta$ & 7 \\
			Reactant deficiency & $\delta_p$ & 4 \\
			\hline
	    \end{tabular}
    \end{center}
\end{table}

Recall that deficiency measures the linear dependence of the network’s reactions, i.e., the higher the deficiency, the higher the linear dependence. The insulin signaling system is of deficiency $\delta = 7$, indicative of high complexity of the network.

\subsubsection{Positive Dependency}

CRNToolbox \cite{CRNToolbox} results show that the metabolic insulin signaling network is positive dependent, implying that a set of rate constants exists for which the mass action system can admit a positive equilibrium.

Proposition \ref{prop:1} shows the key implication of positive dependency in a general form for rate constant-interaction map decomposable (RID) kinetic systems (a kinetics $K$ is RID if it assigns to each $i$th reaction a function $K_i (x) = k_i I_i (X)$ with rate constant $k_i > 0$ and interaction map $I_i (X) \in \mathbb{R}^\mathscr{R}$ for $X \in \mathbb{R}_{\geq 0}^\mathscr{S}$). More information about RID kinetic systems can be found in Nazareno et al \cite{NEML2019} where they were introduced. The proposition with the implication of positive dependency was shown for mass action systems by Feinberg \cite{FEIN1987}.

\begin{proposition}
	\label{prop:1}
	For any RID kinetics $K$ on a positive dependent network $\mathscr{N}$, there are rate constants such that the set of positive equilibria $E_+ (\mathscr{N}, K) \neq \varnothing$.
\end{proposition}

\begin{proof}
	Since $\mathscr{N}$ is positive dependent, then for each reaction $i: C_i \rightarrow C_i'$ there is a positive number $\alpha_i$ such that
	\begin{equation*}
		\sum_i \alpha_i (C_i' - C_i) = 0.
	\end{equation*}
	For a positive vector $X^*$, $I_i (X^*) > 0$ by definition of RID kinetics. Set $\displaystyle k_i^* = \frac{\alpha_i}{I_i (X^*)}$. Then
	\begin{align*}
		f(X^*) &= \sum_i K_i (X^*) (C_i' - C_i) \\
		&= \sum_i k_i^* I_i (X^*) (C_i' - C_i) \\
		&= \sum_i \alpha_i (C_i' - C_i) \\
		&= 0
	\end{align*}
	i.e., $X^* \in E_+ (\mathscr{N}, K)$.
\end{proof}

\begin{remark}
	Proposition \ref{prop:1} is a generalization of Lemma 3.5.3 of Feinberg \cite{FEIN2019}.
\end{remark}

\subsubsection{$t$-minimality}

Table \ref{tab:INSMS} shows that there is only one terminal strong linkage class in each linkage class since $t = 13 = \ell$, i.e., the network is $t$-minimal. Feinberg and Horn \cite{FEHO1977} showed that this implies the coincidence of the kinetic and stoichiometric subspaces.

The noncoincidence of the kinetic and stoichiometric subspaces is closely related to the degeneracy of equilibria. In \cite{FEWI2012}, the authors showed that if the two subspaces differ, then all equilibria are degenerate. In his book, Feinberg describes anomalies that can occur if the two subspaces do not coincide (Section 3.A.1 of \cite{FEIN2019}).

Thus, the $t$-minimality of the metabolic insulin signaling network is a necessary condition for the existence of nondegenerate equilibria as detailed in Remark 4.11 of \cite{FEWI2012}.

\subsubsection{Nonconservativity}

Recall that a reaction network $(\mathscr{S}, \mathscr{C}, \mathscr{R})$ is conservative whenever the orthogonal complement $S^\perp$ of its stoichiometric subspace $S$ contains a strictly positive member of $\mathbb{R}^\mathscr{S}$, i.e., $S^\perp \cap \mathbb{R}_{>0}^\mathscr{S} \neq \varnothing$. Otherwise, the network is called nonconservative.

The metabolic insulin signaling network has $m = 20$ species and rank $s = 15$. Thus, $S^\perp$ has dimension 5. A basis for $S^\perp$ is the set
\begin{align*}
	\{	& [0, 0, 0, 0, 0, 0, 0, 1, 1, -1, 0, 0, 0, 0, 0, 0, 0, 0, 0, 0]^\top, \\
		& [0, 0, 0, 0, 0, 0, 0, 0, 0, 1, 1, 0, 0, 0, 0, 0, 0, 0, 0, 0]^\top, \\
		& [0, 0, 0, 0, 0, 0, 0, 0, 0, 0, 0, 1, 1, 1, 0, 0, 0, 0, 0, 0]^\top, \\
		& [0, 0, 0, 0, 0, 0, 0, 0, 0, 0, 0, 0, 0, 0, 1, 1, 0, 0, 0, 0]^\top, \\
		& [0, 0, 0, 0, 0, 0, 0, 0, 0, 0, 0, 0, 0, 0, 0, 1, 1, 0, 0, 0]^\top \}
\end{align*}
which implies that $S^\perp \cap \mathbb{R}_{>0}^\mathscr{S} = \varnothing$. Therefore, the metabolic insulin signaling network is nonconservative. The positive implications of this property for ACR will be discussed in Section \ref{sec:acr}.

\subsubsection{Concordance}

The most important specific result of the CRNT analysis is the network's concordance.

Concordance is closely related to weakly monotonic kinetics as shown by Proposition 4.8 of \cite{SHFE2012}. It shows that, for any weakly monotonic kinetic system, injectivity (and, therefore, the absence of distinct stoichiometrically compatible equilibria, at least one of which is positive) can be precluded merely by establishing concordance of the underlying reaction network. The converse of the said proposition is not true, i.e., there can be a weakly monotonic kinetic system (in fact, a mass action system) that is injective even when its underlying reaction network is not concordant.

Theorem 4.11 of \cite{SHFE2012} shows that the class of concordant networks is precisely the class of networks that are injective for every assignment of a weakly monotonic kinetics.

In summary, concordance implies injectivity (of the mass action kinetics) and hence monostationarity, i.e., there is at most one positive equilibrium in each stoichiometric compatibility class. We ran the Concordance Test in the CRNToolbox which showed that the metabolic insulin signaling network is concordant. Running the Mass Action Injectivity Test and the Higher Deficiency Test confirms that the kinetics of the insulin signaling system is injective and that the system is monostationary, respectively.

\section{Independent Decompositions of the Insulin Signaling System}
\label{sec:indepDecomp}

In this section, we present the finest nontrivial independent decomposition of the metabolic insulin signaling network. We also present a deficiency-oriented coarsening of the decomposition which will be useful in the analysis of the network's ACR.

\subsection{The Finest Nontrivial Independent Decomposition}
\label{subsec:finest}

Hernandez and De la Cruz \cite{HEDC2021} provide an algorithm for constructing an independent decomposition. Their procedure utilizes a coordinate graph that represents the network using the reaction vectors of the CRN. A nontrivial independent decomposition is generated if the coordinate graph is not connected and, in this case, each connected component of the graph constitutes a partition of the set of reaction vectors of the CRN. Following their procedure, we developed a Matlab code to run the said algorithm and got the independent decomposition $\mathscr{N} = \{R_1, \ldots, R_{35}\}$ consisting of 10 subnetworks:
\begin{multicols}{2}
\noindent
\begin{align*}
	& \mathscr{N}_1 = \{ R_1, \ldots, R_{12}, R_{15}, R_{16} \} \\
	& \mathscr{N}_2 = \{ R_{13}, R_{14} \} \\
	& \mathscr{N}_3 = \{ R_{17}, R_{18}, R_{19} \} \\
	& \mathscr{N}_4 = \{ R_{20}, R_{21} \} \\
	& \mathscr{N}_5 = \{ R_{22}, R_{23} \} \\
	& \mathscr{N}_6 = \{ R_{24}, R_{25} \} \\
	& \mathscr{N}_7 = \{ R_{26}, R_{27} \} \\
	& \mathscr{N}_8 = \{ R_{28}, R_{29} \} \\
	& \mathscr{N}_9 = \{ R_{30}, \ldots, R_{33} \} \\
	& \mathscr{N}_{10} = \{ R_{34}, R_{35} \}.
\end{align*}
\end{multicols}
Furthermore, the independent decomposition obtained is precisely the finest independent decomposition of $\mathscr{N}$ \cite{HEAM2022}. The network numbers of the subnetworks are presented in Table \ref{tab:indepDecomp}.

\begin{table}[ht]
    \begin{center}
	    \caption{Network numbers of the subnetworks $\mathscr{N}_i$ of the finest independent decomposition of the metabolic insulin signaling network $\mathscr{N}$}
	    \label{tab:indepDecomp}
	    \begin{tabular}{@{}llcccccccccc@{}}
	        \hline
			Network numbers & & $\mathscr{N}_1$ & $\mathscr{N}_2$ & $\mathscr{N}_3$ & $\mathscr{N}_4$ & $\mathscr{N}_5$ & $\mathscr{N}_6$ & $\mathscr{N}_7$ & $\mathscr{N}_8$ & $\mathscr{N}_9$ & $\mathscr{N}_{10}$ \\
			\hline
			Species & $m$ & 7 & 1 & 4 & 3 & 3 & 2 & 3 & 3 & 4 & 1 \\
			Complexes & $n$ & 7 & 2 & 6 & 2 & 4 & 2 & 4 & 4 & 6 & 2 \\
			Reactant complexes & $n_r$ & 7 & 2 & 3 & 2 & 2 & 2 & 2 & 2 & 4 & 2 \\
			Reversible reactions & $r_\text{rev}$ & 5 & 1 & 0 & 1 & 0 & 1 & 0 & 0 & 1 & 1 \\
			Irreversible reactions & $r_\text{irrev}$ & 4 & 0 & 3 & 0 & 2 & 0 & 2 & 2 & 2 & 0 \\
			Reactions & $r$ & 14 & 2 & 3 & 2 & 2 & 2 & 2 & 2 & 4 & 2 \\
			Linkage classes & $\ell$ & 1 & 1 & 3 & 1 & 2 & 1 & 2 & 2 & 3 & 1 \\
			Strong linkage classes & $s \ell$ & 1 & 1 & 6 & 1 & 4 & 1 & 4 & 4 & 5 & 1 \\
			Terminal strong linkage classes & $t$ & 1 & 1 & 3 & 1 & 2 & 1 & 2 & 2 & 3 & 1 \\
			Rank & $s$ & 6 & 1 & 1 & 1 & 1 & 1 & 1 & 1 & 1 & 1 \\
			Reactant rank & $q$ & 7 & 1 & 3 & 2 & 2 & 2 & 2 & 2 & 4 & 1 \\
			Deficiency & $\delta$ & 0 & 0 & 2 & 0 & 1 & 0 & 1 & 1 & 2 & 0 \\
			Reactant deficiency & $\delta_p$ & 0 & 1 & 0 & 0 & 0 & 0 & 0 & 0 & 0 & 1 \\
			\hline
	    \end{tabular}
	\end{center}
\end{table}

Consider the coarsening $\mathscr{N} = \mathscr{N}_1^* \cup \mathscr{N}_2^* \cup \mathscr{N}_3^*$ where $\mathscr{N}_1^* = \mathscr{N}_1 \cup \mathscr{N}_2$, $\mathscr{N}_2^* = \mathscr{N}_3 \cup \ldots \cup \mathscr{N}_8$, and $\mathscr{N}_3^* = \mathscr{N}_9 \cup \mathscr{N}_{10}$. This decomposition shows the three subsystems considered by Sedaghat et al in the construction of their system: the insulin receptor binding and recycling subsystem ($\mathscr{N}_1^*$), the postreceptor signaling subsystem ($\mathscr{N}_2^*$), and the GLUT4 translocation subsystem ($\mathscr{N}_3^*$). It is significant to note that CRNT analysis using independent decomposition reveals that these subnetworks of the insulin signaling system are not just functionally but also structurally important, i.e., positive equilibria of the whole system come from equilibria of these systems.

Recall that a decomposition is said to be bi-independent if it is both independent and incidence independent.

\begin{proposition}
    \label{prop:2}
	The decomposition $\mathscr{N} = \mathscr{N}_1 \cup \ldots \cup \mathscr{N}_{10}$ is bi-independent.
\end{proposition}

\begin{proof}
	Table \ref{tab:indepDecomp} shows that $\delta = \delta_1 + \ldots + \delta_{10}$ where $\delta_i$ is the deficiency of subnetwork $\mathscr{N}_i$. Hence, by Proposition 8 of \cite{FML2021}, the decomposition is bi-independent.
\end{proof}

Note that since Proposition 7 of \cite{FML2021} implies that every coarsening is also incidence independent, then every coarsening is bi-independent as well.

\subsection{Deficiency-Oriented Coarsening}
\label{sec:def}

We next consider a deficiency-oriented coarsening of the independent decomposition. The deficiency-oriented coarsening is $\mathscr{N} = \mathscr{N}_A \cup \mathscr{N}_B$ where $\mathscr{N}_A = \mathscr{N}_1 \cup \mathscr{N}_2 \cup \mathscr{N}_4 \cup \mathscr{N}_6 \cup \mathscr{N}_{10}$ and $\mathscr{N}_B = \mathscr{N}_3 \cup \mathscr{N}_5 \cup \mathscr{N}_7 \cup \mathscr{N}_8 \cup \mathscr{N}_9$. Let $\mathscr{S}_i$, $\mathscr{C}_i$, $K_i$, and $S_i$ be the set of species, set of complexes, kinetics, and stoichiometric subspace, respectively, of $\mathscr{N}_i$ for $i = A, B$.

In $\mathscr{N}_A$, we consider together all subnetworks of the finest independent decomposition which have a deficiency of 0. On the other hand, we put together all nonzero-deficiency subnetworks in $\mathscr{N}_B$. The network numbers of $\mathscr{N}_A$ and $\mathscr{N}_B$ are presented in Table \ref{tab:defOriented}. Note that the set of common species of the subnetworks is $\mathscr{S}_A \cap \mathscr{S}_B = \{ X_4, X_5, X_{10}, X_{12}, X_{13}, X_{20} \}$ while its set of common complexes $\mathscr{C}_A \cap \mathscr{C}_B = \{ X_{13}, X_{20} \}$.

\begin{table}[ht]
    \begin{center}
	    \caption{Network numbers of the subnetworks of the deficiency-oriented coarsening of the metabolic insulin signaling network $\mathscr{N} = \mathscr{N}_A \cup \mathscr{N}_B$}
	    \label{tab:defOriented}
	    \begin{tabular}{@{}llcc@{}}
	        \hline
			Network numbers & & $\mathscr{N}_A$ & $\mathscr{N}_B$ \\
			\hline
			Species & $m$ & 13 & 13 \\
			Complexes & $n$ & 13 & 24 \\
			Reactant complexes & $n_r$ & 13 & 13 \\
			Reversible reactions & $r_\text{rev}$ & 9 & 1 \\
			Irreversible reactions & $r_\text{irrev}$ & 4 & 11 \\
			Reactions & $r$ & 22 & 13 \\
			Linkage classes & $\ell$ & 3 & 12 \\
			Strong linkage classes & $s \ell$ & 3 & 23 \\
			Terminal strong linkage classes & $t$ & 3 & 12 \\
			Rank & $s$ & 10 & 5 \\
			Reactant rank & $q$ & 12 & 11 \\
			Deficiency & $\delta$ & 0 & 7 \\
			Reactant deficiency & $\delta_p$ & 1 & 2 \\
			\hline
	    \end{tabular}
    \end{center}
\end{table}

The two subnetworks contrast in further properties besides deficiency: $\mathscr{N}_A$ is weakly reversible and nonconservative, while $\mathscr{N}_B$ is not weakly reversible and conservative. However, they are both $t$-minimal and concordant (hence monostationary). The Deficiency Zero Theorem for mass action systems, in fact, guarantees that $\mathscr{N}_A$ has a unique complex balanced equilibrium in every stoichiometric compatibility class. Because $\mathscr{N}_B$ is not weakly reversible, it has no complex balanced equilibria and there may be stoichiometric compatibility classes without an equilibrium.

Lemma \ref{lem:3} paves the way for the usefulness of the deficiency-oriented decomposition.

\begin{lemma}
	\label{lem:3}
	Let $\mathscr{N} = (\mathscr{S}, \mathscr{C}, \mathscr{R})$ be a CRN. If $\mathscr{N} = \mathscr{N}_1 \cup \mathscr{N}_2$ is an independent decomposition of $\mathscr{N}$ with $\mathscr{S}_i$, $S_i$, and $K_i$ the set of species, stoichiometric subspace, and kinetics, respectively, of $\mathscr{N}_i$ for $i = 1, 2$, then:
	\begin{enumerate}[$(i)$]
		\item If $\mathscr{S}_1 \cap \mathscr{S}_2 = \varnothing$, then for $i = 1, 2$, the projection maps $p_i: \mathbb{R}^\mathscr{S} \rightarrow \mathbb{R}^{\mathscr{S}_i}$ induce an isomorphism of $\mathbb{R}^\mathscr{S} / S \rightarrow \mathbb{R}^{\mathscr{S}_1} / S_1 \times \mathbb{R}^{\mathscr{S}_2} / S_2$ where $\mathscr{S}_1 \cup \mathscr{S}_2 = \mathscr{S}$; and
		\item If $\mathscr{S}_1 \cap \mathscr{S}_2 \neq \varnothing$, let $p_{\mathscr{CS}}: \mathbb{R}^\mathscr{S} \rightarrow \mathbb{R}^\mathscr{CS}$ be the projection to $\mathscr{CS} := \mathscr{S}_1 \cap \mathscr{S}_2$. Then $x \in E_+ (\mathscr{N}_i, K_i)$ is in $E_+ (\mathscr{N}, K)$ if and only if $p_\mathscr{CS} (x) \in p_\mathscr{CS} (E_+ (\mathscr{N}_j, K_j))$, $j \neq i$.
	\end{enumerate}
\end{lemma}

\begin{proof}
	\textcolor{white}{}
	
	$(i)$ Let $p(x + S) = (p_1 (x) + S_1, p_2 (x) + S_2)$. It is well-defined because we have the following: $x - x' \in S \Rightarrow p_i (x) - p_i (x') = p_i (x - x') \in S_i$. Since $\mathscr{S}_1 \cap \mathscr{S}_2 = \varnothing$, $x = p_1 (x) + p_2 (x) \in S$ if the projections are in $S_1$ and $S_2$, respectively, showing the injectivity. Furthermore, the number of species $m_1 + m_2 = m$ (since $\mathscr{S}_1 \cap \mathscr{S}_2 = \varnothing$) and the rank $s_1 + s_2 = s$ (due to direct sum). Hence, the domain and codomain have the same dimensions, showing the isomorphism.
	
	$(ii)$ This follows directly from Feinberg's Decomposition Theorem that if a decomposition is independent, then $E_+ (\mathscr{N}, K) = E_+ (\mathscr{N}_1, K_1) \cap E_+ (\mathscr{N}_2, K_2)$.
\end{proof}

\begin{remark}
	\textcolor{white}{}
	\begin{enumerate}
		\item Statement $(i)$ of Lemma \ref{lem:3} can be generalized to the case of $\mathscr{S}_1 \cap \mathscr{S}_2 \neq \varnothing$ but the isomorphism obtained is not directly relevant to our equilibria analysis. Hence, we have relegated it to Appendix \ref{app:lem}.
		\item Although statement $(ii)$ of Lemma \ref{lem:3} is an easy consequence of Feinberg's Decomposition Theorem, it turns out to be useful in cases where at least one of the subnetworks has a well-understood set of positive equilibria, as in the case of the metabolic insulin signaling network. This is shown in Proposition \ref{prop:4} and Example \ref{ex:1}.
	\end{enumerate}
\end{remark}

\begin{proposition}
    \label{prop:4}
	Let $(\mathscr{N}, K)$ be the metabolic insulin signaling network and $\mathscr{N} = \mathscr{N}_A \cup \mathscr{N}_B$ its deficiency-oriented decomposition. Then:
	\begin{enumerate}[$(i)$]
		\item The map $\epsilon: E_+ (\mathscr{N}, K) \rightarrow \mathbb{R}^\mathscr{S}/S$ given by $\epsilon(x) := x + S$ is injective.
		\item The map $\epsilon_A: E_+ (\mathscr{N}_A, K_A) \rightarrow \mathbb{R}^\mathscr{S}/S$ given by $\epsilon_A (x) := x + S$ is surjective.
		\item $Az + S \in \text{Im}(\epsilon)$ if and only if there is $x \in \epsilon_A^{-1} (z + S)$ such that $p_\mathscr{CS} (x) \in p_\mathscr{CS} (E_+ (\mathscr{N}_B, K_B))$.
	\end{enumerate}
\end{proposition}

\begin{proof}
	$(i)$ is equivalent to the statement that $(\mathscr{N}, K)$ is monostationary. $(ii)$ follows from the fact that the map is the composition of $\tilde{\epsilon}_A: E_+ (\mathscr{N}, K) \rightarrow \mathbb{R}^\mathscr{S}/S_A$ (surjective due to the Deficiency Zero Theorem) and the surjective map $x + S_A \rightarrow x + S$. Finally, $(iii)$ is the application of statement $(ii)$ of Lemma \ref{lem:3} to the metabolic insulin signaling network.
\end{proof}

\newpage
\begin{example}
	\label{ex:1}
	The following example illustrates the use of statement $(iii)$ of Proposition \ref{prop:4}.
	
	First, we express the positive equilibria of $\mathscr{N}_A$ in terms of the rate constants:
	\begin{multicols}{2}
	\noindent
	\begin{align*}
		& X_2 = \frac{BCDF - ACEG - BCDF + ACDH}{A^2 E G + ABDF - A^2 D H} \\
		& X_3 = \frac{k_1}{k_2 + k_5} X_2 \\
		& X_4 = \frac{CDEF}{ADEG + B D^2 F - A D^2 H} \\
		& X_5 = \frac{CDF}{ADH - BDF - AEG} \\
		& X_6 = \frac{k_{13}}{k_{14}} \\
		& X_7 = \frac{k_9}{k_{10} + k_{15}} X_4 \\
		& X_8 = \frac{k_{11}}{k_{12} + k_{16}} X_5 \\
		& X_{12} = \frac{k_{20}}{k_{21}} X_{10} X_{11} \\
		& X_{15} = \frac{k_{25}}{k_{24}} X_{13} \\
		& X_{20} = \frac{k_{34}}{k_{35}}
	\end{align*}
	\columnbreak
	\begin{align*}
		& A = \frac{k_1 k_2}{k_2 + k_5} - k_1 - k_7 \\
		& B = k_6 \\
		& C = \frac{k_8 k_{13}}{k_{14}} \\
		& D = \frac{k_9 k_{10}}{k_{10} + k_{15}}  - k_4 - k_9 \\
		& E = k_3 \\
		& F = \frac{k_1 k_5}{k_2 + k_5} \\
		& G = k_4 \\
		& H = \frac{k_{11} k_{12}}{k_{12} + k_{16}} - k_3 - k_6 - k_{11}.
	\end{align*}
	\end{multicols}
	Observe that in subnetwork $\mathscr{N}_A$, the equilibria of $X_{10}$, $X_{11}$, and $X_{13}$ are independent of the rate constants. Note also that there are restrictions on the rate constant values since the $X_i$'s have to be positive. Using the rate constants from \cite{SEDA2002}, we obtain a positive equilibrium for $\mathscr{N}_A$:
	\begin{multicols}{2}
		\noindent
		\begin{align*}
			& X_2 = 0.3700 \\
			& X_3 = 8.8788 \times 10^{-4} \\
			& X_4 = 3.2844 \\
			& X_5 = 10.9491 \\
			& X_6 = 10.0 \\
			& X_7 = 0.0150 \\
			& X_8 = 0.0499 \\
			& X_{10} = 7.0929 \times 10^{-27} \\
			& X_{11} = 0.00105 \\
			& X_{12} = 5.2614 \times 10^{-19} \\
			& X_{13} = 0.3100 \\
			& X_{15} = 0.2900 \\
			& X_{20} = 96.
		\end{align*}
	\end{multicols}
	
	Next, we substitute the values of the species common to $\mathscr{N}_A$ and $\mathscr{N}_B$ into the ODEs for $\mathscr{N}_B$ and look for a positive solution (this implements statement $(iii)$ of Proposition \ref{prop:4}):
	\begin{align*}
		& \dot{X}_9 = k_{19} X_{10} - k_{17} X_9 X_4 - k_{18} X_9 X_5 \\
		& \dot{X}_{10} = k_{17} X_9 X_4 + k_{18} X_9 X_5 - k_{19} X_{10} \\
		& \dot{X}_{13} = k_{22} X_{12} X_{14} - k_{23} X_{13} \\
		& \dot{X}_{14} = k_{23} X_{13} - k_{22} X_{12} X_{14} \\
		& \dot{X}_{16} = k_{27} X_{17} - k_{26} X_{13} X_{16} \\
		& \dot{X}_{17} = k_{26} X_{13} X_{16} - k_{27} X_{17} \\
		& \dot{X}_{18} = k_{29} X_{19} - k_{28} X_{13} X_{18} \\
		& \dot{X}_{19} = k_{28} X_{13} X_{18} - k_{29} X_{19} \\
		& \dot{X}_{20} = k_{31} X_{21} - k_{30} X_{20} - k_{32} X_{17} X_{20} - k_{33} X_{19} X_{20} \\
		& \dot{X}_{21} = k_{30} X_{20} + k_{32} X_{17} X_{20} + k_{33} X_{19} X_{20}  - k_{31} X_{21}
	\end{align*}
	Solving $\dot{X} = 0$, we get the following equilibrium for the other species in $\mathscr{N}_B$:
	\begin{align*}
		& X_9 = 1.5 \times 10^{-40} \\ 
		& X_{14} = 99.3 \\
		& X_{17} = (5.1269 \times 10^{-9}) X_{16} \\ 
		& X_{19} = (5.1269 \times 10^{-9}) X_{18} \\ 
		& X_{21} = (6.7676 \times 10^{-9}) X_{16} + (2.7070 \times 10^{-8}) X_{18} + 4.0 
	\end{align*}
	Note that the equilibrium values of $X_{17}$, $X_{19}$, and $X_{21}$ depend on any values of $X_{16}$ and $X_{18}$. If we (randomly) choose $X_{16} = 99.9$ and $X_{18} = 99.9$, we get the following:
	\begin{align*}
		& X_{17} = 5.1218 \times 10^{-7} \\
		& X_{19} = 5.1218 \times 10^{-7} \\
		& X_{21} = 4.0000.
	\end{align*}
\end{example}

\begin{remark}
	Example \ref{ex:1} provides an alternative method for solving for positive equilibria using smaller subnetworks instead of the (more complex) entire network.
\end{remark}

\section{ACR of Species in Insulin Signaling}
\label{sec:acr}

ACR denotes the invariance of the concentrations of a species at all positive equilibria of a kinetic system. Shinar and Feinberg introduced the concept in 2010 \cite{SHFE2010} and, from experimental observations in \textit{Escherichia coli} subsystems, extracted sufficient (mathematical) conditions for the property. In this section, after a brief review of relevant results on ACR, we present an analysis of the property in the insulin signaling system.

\subsection{Absolute Concentration Robustness}

A pair of reactant complexes $\{ C, C' \}$ is a \textbf{Shinar-Feinberg pair} (SF-pair) \textbf{in species $X$} if their kinetic order vectors (i.e., the corresponding row in the kinetic order matrix) differ only in $X$. In mass action systems, the stoichiometric coefficients of the complexes play the role of the kinetic order vectors.

An SF-pair is called \textbf{nonterminal} if both complexes are not in terminal strong linkage classes. The pair is said to be \textbf{linked} if both complexes are in a linkage class.

Proposition 5.7 of \cite{LLMM2022} presents a framework for ACR using SF-pairs. A special case of condition $(ii)$ of the said proposition is any weakly reversible power law system with reactant-determined kinetics, zero deficiency, and a linked SF-pair in $X$. A proof of ACR in $X$ in this case was already provided in Theorem 6 in the Appendix of \cite{FOME2021}. We refer the reader to \cite{LLMM2022} for the detailed discussion of positive equilibria log-parametrized systems and a proof of the general framework. We apply Proposition 5.7 of \cite{LLMM2022} only in the special case of the deficiency zero subnetwork of the metabolic insulin signaling network.

We will often call subnetworks with the property $(i)$ or $(ii)$ of Proposition 5.7 of \cite{LLMM2022} ``(low deficiency) ACR building blocks''. A computational approach to the framework was developed in \cite{FOMF2021}. For the ACR analysis, we implemented this algorithm using Matlab.

Another theorem that we will apply to the ACR analysis of the metabolic insulin signaling network is Theorem 5.5 of \cite{MEST2021} regarding the stable ACR for one-dimensional networks which we will refer to as the ``Meshkat et al criterion''. As defined by Meshkat et al \cite{MEST2021}, a kinetic system has stable ACR in a species $X$ (for a set of rate constants) if all the steady states of the kinetic system (for the set of rate constants) are stable.

\begin{remark}
	The sufficient conditions for ACR in a species in the works of Lao et al \cite{LLMM2022} and Meshkat et al \cite{MEST2021} establish the property for all rate constants for which the system has a positive equilibrium. Meshkat et al have proposed the convention of ``vacuous ACR'' for the case in which no positive equilibrium exists: in such a scenario, ACR in all species is assumed. This enables the more convenient terminology that the above sufficient conditions provide ACR ``for all rate constants'' and we adopt this terminology in the following analysis.
\end{remark}

\subsection{ACR Analysis of Rank 1 Subnetworks of the Finest Independent Decomposition}

The finest independent decomposition using the Hernandez-De la Cruz algorithm consists of one subnetwork of rank 6 ($\mathscr{N}_1$) and 9 subnetworks of rank 1 ($\mathscr{N}_2, \ldots, \mathscr{N}_{10}$) (see Table \ref{tab:indepDecomp}). It is natural to first try to apply the ACR criterion of Meshkat et al to the latter set of subnetworks.

Condition $2(b)$ of the Meshkat et al criterion says that all reactions, taken pairwise, must be SF-pairs in species $X$. Table \ref{tab:meshkat} shows the result of the verification.

\begin{table}[ht]
    \begin{center}
	    \caption{Verification of condition $2(b)$ of the Meshkat et al criterion}
	    \label{tab:meshkat}
	    \begin{tabular}{@{}|c|c|c|c|@{}}
	        \hline
			Subnetwork & Reactions & Non-SF-pair, e.g. & $2(b)$ satisfied for species $X_i$ \\
			\hline
			$\mathscr{N}_2$ & $R_{13}$, $R_{14}$ & none & Yes for $X_6$ \\
			\hline
			$\mathscr{N}_3$ & $R_{17}$, $R_{18}$, $R_{19}$ & $R_{17}$, $R_{18}$ & No \\
			\hline
			$\mathscr{N}_4$ & $R_{20}$, $R_{21}$ & $R_{20}$, $R_{21}$ & No \\
			\hline
			$\mathscr{N}_5$ & $R_{22}$, $R_{23}$ & $R_{22}$, $R_{23}$ & No \\
			\hline
			$\mathscr{N}_6$ & $R_{24}$, $R_{25}$ & $R_{24}$, $R_{25}$ & No \\
			\hline
			$\mathscr{N}_7$ & $R_{26}$, $R_{27}$ & $R_{26}$, $R_{27}$ & No \\
			\hline
			$\mathscr{N}_8$ & $R_{28}$, $R_{29}$ & $R_{28}$, $R_{29}$ & No \\
			\hline
			$\mathscr{N}_9$ & $R_{30}$, $R_{31}$, $R_{32}$, $R_{33}$ & $R_{30}$, $R_{31}$ & No \\
			\hline
			$\mathscr{N}_{10}$ & $R_{34}$, $R_{35}$ & none & Yes for $X_{20}$ \\
			\hline
	    \end{tabular}
    \end{center}
\end{table}

Subnetworks $\mathscr{N}_2$ and $\mathscr{N}_{10}$ satisfy condition $2(a)$ of the Meshkat et al criterion as well since they each consist of a reversible pair and are identical to the required embedded network. Hence, $X_6$ and $X_{20}$ have ACR for all rate constants in the whole network since the decomposition is independent (see Proposition 4.4 of \cite{LLMM2022}).

\begin{remark}
	\textcolor{white}{}
	\begin{enumerate}
		\item The Meshkat et al criterion ensures stability of the equilibria only within the subnetworks and not for the whole network so we cannot infer any claim about the property for $X_6$ and $X_{20}$.
		\item In a similar vein, non-ACR for a species in the subnetwork does not necessarily imply non-ACR in the whole network since the set of positive equilibria of the latter is generally smaller than that of the former.
	\end{enumerate}
\end{remark}

\subsection{ACR Analysis of Subnetworks in the Deficiency-Oriented Independent Decomposition}
\label{sec:acrDef}

Implementing the algorithm of Fontanil et al \cite{FOMF2021} and using Proposition 5.7 of \cite{LLMM2022}, we discover Shinar-Feinberg reaction pairs in appropriate low-deficiency subnetworks of coarsenings of the finest independent decomposition in Section \ref{subsec:finest}. We find that the system has ACR in 8 species (from the zero-deficiency building blocks): $X_2$, $X_3$, $X_4$, $X_5$, $X_6$, $X_7$, $X_8$, $X_{20}$. The other species are not identified as having ACR because the reactant complexes of their associated SF-pairs are not nonterminal in the subnetwork generated by the building block. Table \ref{tab:ACR} provides an overview.

It is particularly significant that the system's output to glucose energy metabolism \linebreak ($X_{20} = \text{intracellular GLUT4}$) has ACR. GLUT4 is a key transporter of glucose into neurons. Under healthy conditions (i.e., no insulin resistance), the insulin signaling system works such that GLUT4, which determines the amount of glucose transported into a cell, is held constant. GLUT4, coupled with adequate glucose supply, enables reliable cellular energy production. Keeping the value of GLUT4 is very important for glucose energy metabolism, showing real robustness in the system. Energy processing of neurons works properly due to this robustness.

\begin{table}[ht!]
    \begin{center}
	    \caption{Species with ACR}
	    \label{tab:ACR}
	    \begin{tabular}{@{}|c|c|c|c|c|@{}}
	        \hline
			Species & SF-pair & ACR building block & Deficiency & Comments \\
			\hline
			$X_2$ & $R_1$, $R_{13}$ & $\mathscr{N}_A$ & 0 & linked SF-pair \\
			 & $R_1$, $R_{34}$ & $\mathscr{N}_A$ & 0 & linked SF-pair \\
			 & $R_7$, $R_{13}$ & $\mathscr{N}_A$ & 0 & linked SF-pair \\
			 & $R_7$, $R_{34}$ & $\mathscr{N}_A$ & 0 & linked SF-pair \\
			\hline
			$X_3$ & $R_2$, $R_{13}$ & $\mathscr{N}_A$ & 0 & linked SF-pair \\
			 & $R_2$, $R_{34}$ & $\mathscr{N}_A$ & 0 & linked SF-pair \\
			 & $R_5$, $R_{13}$ & $\mathscr{N}_A$ & 0 & linked SF-pair \\
			 & $R_5$, $R_{35}$ & $\mathscr{N}_A$ & 0 & linked SF-pair \\
			\hline
			$X_4$ & $R_4$, $R_{13}$ & $\mathscr{N}_A$ & 0 & linked SF-pair \\
			 & $R_4$, $R_{34}$ & $\mathscr{N}_A$ & 0 & linked SF-pair \\
			 & $R_9$, $R_{13}$ & $\mathscr{N}_A$ & 0 & linked SF-pair \\
			 & $R_9$, $R_{34}$ & $\mathscr{N}_A$ & 0 & linked SF-pair \\
			\hline
			$X_5$ & $R_3$, $R_{13}$ & $\mathscr{N}_A$ & 0 & linked SF-pair \\
			 & $R_3$, $R_{34}$ & $\mathscr{N}_A$ & 0 & linked SF-pair \\
			 & $R_6$, $R_{13}$ & $\mathscr{N}_A$ & 0 & linked SF-pair \\
			 & $R_6$, $R_{34}$ & $\mathscr{N}_A$ & 0 & linked SF-pair \\
			 & $R_{11}$, $R_{13}$ & $\mathscr{N}_A$ & 0 & linked SF-pair \\
			 & $R_{11}$, $R_{34}$ & $\mathscr{N}_A$ & 0 & linked SF-pair \\
			\hline
			$X_6$ & $R_8$, $R_{13}$ & $\mathscr{N}_A$ & 0 & linked SF-pair \\
			 & $R_8$, $R_{34}$ & $\mathscr{N}_A$ & 0 & linked SF-pair \\
			 & $R_{13}$, $R_{14}$ & $\mathscr{N}_A$ & 0 & linked SF-pair \\
			 & $R_{14}$, $R_{34}$ & $\mathscr{N}_A$ & 0 & linked SF-pair \\
			\hline
			$X_7$ & $R_{10}$, $R_{13}$ & $\mathscr{N}_A$ & 0 & linked SF-pair \\
			 & $R_{10}$, $R_{34}$ & $\mathscr{N}_A$ & 0 & linked SF-pair \\
			 & $R_{13}$, $R_{15}$ & $\mathscr{N}_A$ & 0 & linked SF-pair \\
			 & $R_{15}$, $R_{34}$ & $\mathscr{N}_A$ & 0 & linked SF-pair \\
			\hline
			$X_8$ & $R_{12}$, $R_{13}$ & $\mathscr{N}_A$ & 0 & linked SF-pair \\
			 & $R_{12}$, $R_{34}$ & $\mathscr{N}_A$ & 0 & linked SF-pair \\
			 & $R_{13}$, $R_{16}$ & $\mathscr{N}_A$ & 0 & linked SF-pair \\
			 & $R_{16}$, $R_{34}$ & $\mathscr{N}_A$ & 0 & linked SF-pair \\
			\hline
			$X_{13}$ & $R_{13}$, $R_{23}$ & $\mathscr{N}_2 \cup \mathscr{N}_5$ & 1 &  \\
			 & $R_{23}$, $R_{34}$ & $\mathscr{N}_5 \cup \mathscr{N}_{10}$ & 1 &  \\
			\hline
			$X_{14}$ & $R_{21}$, $R_{22}$ & $\mathscr{N}_4 \cup \mathscr{N}_5$ & 1 &  \\
			\hline
			$X_{16}$ & $R_{25}$, $R_{26}$ & $\mathscr{N}_6 \cup \mathscr{N}_7$ & 1 &  \\
			\hline
			$X_{17}$ & $R_{13}$, $R_{27}$ & $\mathscr{N}_2 \cup \mathscr{N}_7$ & 1 &  \\
			 & $R_{27}$, $R_{34}$ & $\mathscr{N}_7 \cup \mathscr{N}_{10}$ & 1 &  \\
			\hline
			$X_{18}$ & $R_{25}$, $R_{28}$ & $\mathscr{N}_6 \cup \mathscr{N}_8$ & 1 &  \\
			\hline
			$X_{19}$ & $R_{13}$, $R_{29}$ & $\mathscr{N}_2 \cup \mathscr{N}_8$ & 1 &  \\
			 & $R_{29}$, $R_{34}$ & $\mathscr{N}_8 \cup \mathscr{N}_{10}$ & 1 &  \\
			\hline
			$X_{20}$ & $R_{13}$, $R_{35}$ & $\mathscr{N}_A$ & 0 & linked SF-pair \\
			 & $R_{34}$, $R_{35}$ & $\mathscr{N}_A$ & 0 & linked SF-pair \\
			\hline
	    \end{tabular}
    \end{center}
\end{table}

\begin{remark}
	The ACR in 8 species of the network is inferred from ACR in the deficiency zero subnetwork $\mathscr{N}_A$. A necessary condition for this occurrence is the nonconservativity of $\mathscr{N}$. If $\mathscr{N}$ were conservative, i.e., there was a positive vector in $S^\perp$, then $S^\perp = (S_A + S_B)^\perp = S_A^\perp \cap S_B^\perp$. Hence, $\mathscr{N}_A$ would be a conservative deficiency zero mass action network and have no ACR in any species (Theorem 9.7.1 of \cite{FEIN2019}).
\end{remark}

\subsection{ACR Analysis of the Metabolic Insulin Signaling Network for a Set of Rate Constants}

Table \ref{tab:ACR} shows that there are SF-pairs in deficiency one subnetworks in further independent coarsenings of the metabolic insulin signaling network. The failure of the three sufficient conditions for ACR for all rate constants (Meshkat et al criterion) led us to suspect that the remaining species (i.e., $X_{13}$, $X_{14}$, $X_{16}$, $X_{17}$, $X_{18}$, and $X_{19}$) do not have this property. The test is done using the set of rate constants available from \cite{SEDA2002}. 

Example \ref{ex:1} in Section \ref{sec:def} presents a positive equilibrium of the metabolic insulin signaling network. It is evident from the example that positive equilibria for the metabolic insulin signaling network have the same value for $X_2, \ldots, X_8$ and $X_{20}$. For the other 12 species, their equilibrium value is dependent on the following variables: $X_{10}$, $X_{11}$, $X_{13}$, $X_{16}$, and $X_{18}$. This suggests that there are infinitely many possible positive equilibria.

Our computations verify that, for the set of rate constants used above, species $X_2, \ldots, X_8$ and $X_{20}$ have stable ACR while the remaining species do not have ACR (since the 5 species identified as ``independent variables'' can take various values and the remaining species will vary according to the variation of the variables).

These computations indeed confirm that species $X_{13}$, $X_{14}$, $X_{16}$, $X_{17}$, $X_{18}$, and $X_{19}$ do not have ACR for all rate constants.

\begin{remark}
	The number of species exhibiting ACR (which we call ``ACR species'' for short) is an inverse measure of the variation in the equilibria composition: the more ACR species there are, the less the variation. An extreme case is ACR in all species, which is equivalent to the system having a unique equilibrium (in the entire species space). With 8 ACR species among 20, the metabolic insulin signaling network has a fairly high variation in equilibria composition. The deficiency-oriented decomposition reveals the source of this variation. The deficiency zero subnetwork $\mathscr{N}_A$ contains all the ACR species. Furthermore, since the insulin signaling system is a log-parametrized system, it follows from the results of Lao et al \cite{LLMM2022} that the number of ACR species (among the total 13 species in the subnetwork) is bounded by the subnetwork's rank ($s_A = 10$). This shows that the high variation in equilibria composition is caused primarily by the  lack of ACR species among the 13 species in the deficiency 7 subnetwork $\mathscr{N}_B$.
\end{remark}

\section{Summary and Outlook}
\label{sec:conclusion}

The insulin signaling system is an important metabolic system. In this study, we derived a CRN of the Sedaghat et al insulin signaling model with 20 species, 35 complexes, and 35 reactions. We have shown that it is a nonconservative, non-weakly reversible, and high deficiency ($\delta = 7$) system. The positive dependence of the reaction network ensures the existence of rate constants under which the mass action system has positive equilibria (Proposition \ref{prop:1}). Additionally, the network's $t$-minimality implies that the kinetic and stoichiometric subspaces coincide, which is necessary for the existence of nondegenerate equilibria. Moreover, the network is concordant, which implies that the system's species formation rate function, when restricted to any stoichiometric compatibility class, is injective. It follows, then, that the kinetic system is monostationary, i.e., there is at most one positive equilibrium in each stoichiometric compatibility class.

We obtained the finest independent decomposition of the metabolic insulin signaling network $\mathscr{N} = \{R_1, \ldots, R_{35}\}$ consisting of 10 subnetworks which we have shown to be bi-independent (Proposition \ref{prop:2}). CRNT analysis using a coarsening of the decomposition revealed three subnetworks of the metabolic insulin signaling network which are not only functionally but also structurally important. Upon considering a deficiency-oriented coarsening of the finest decomposition, we have shown how a binary decomposition can be viewed in relation to a network's set of positive equilibria (Lemma \ref{lem:3}). We have also developed a method of determining positive equilibria of the metabolic insulin signaling network using its deficiency-oriented coarsening (Proposition \ref{prop:4}). This provides an alternative method for solving positive equilibria analytically using smaller subnetworks instead of the (more complex) entire network.

For the ACR analysis, we have shown that subnetworks $\mathscr{N}_2$ and $\mathscr{N}_{10}$ satisfy the Meshkat et al criterion, but species $X_6$ and $X_{20}$ (which come from the said subnetworks) have ACR for all rate constants only in their respective subnetwork. This implies that the stability of the equilibria is only within the subnetworks and not for the whole network. Upon implementing the algorithm of Fontanil et al \cite{FOMF2021}, however, we found that the system had ACR in 8 species. This restricts the variability of the positive equilibria and also suggests that this ``structural source of robustness'' may be an important factor in the system's overall robustness. However, overall there is still a high variation in equilibria composition due to the lack of ACR species in the deficiency 7 subnetwork of the network's deficiency-oriented coarsening. For the rate constants used in the study of Sedaghat et al \cite{SEDA2002}, we have verified that species $X_2, \ldots, X_8$, and $X_{20}$ indeed have stable ACR. Interestingly, $X_{20}$, i.e., the insulin-regulated glucose transporter GLUT4, plays an important role in glucose energy metabolism.

Our analysis of the Sedaghat et al model is the first part of a two-step research effort on metabolic insulin signaling. In their paper, Sedaghat et al constructed the model from data of healthy cells. As a second step, we have started a reaction network analysis of a model of metabolic insulin signaling by Br{\"a}nnmark et al \cite{BNFBECS2013} based on cell data from type 2 diabetes patients, i.e., cells with insulin resistance. Preliminary results already indicate significant differences to our main results on the Sedaghat et al model \cite{LLM2022}.



\begin{appendices}

\counterwithin{example}{section}

\section{Notations and Definition of Terms}

In this section, we lay the foundation of Chemical Reaction Network Theory (CRNT) by discussing the definition of terms used in the paper. After discussing the fundamentals of chemical reaction networks and kinetic systems, we review important terminologies related to decomposition theory.

\subsection{Chemical Reaction Networks}

A \textbf{chemical reaction network} (CRN) $\mathscr{N}$ is a triple $(\mathscr{S}, \mathscr{C}, \mathscr{R})$ of nonempty finite sets $\mathscr{S}$, $\mathscr{C}$, and $\mathscr{R}$ of $m$ species, $n$ complexes, and $r$ reactions, respectively, where $\mathscr{C} \subseteq \mathbb{R}_{\geq 0}^{\mathscr{S}}$ and $\mathscr{R} \subset \mathscr{C} \times \mathscr{C}$ satisfying the following:
\begin{enumerate}[($i$)]
	\item $(C, C) \notin \mathscr{R}$ for any $C \in      \mathscr{C}$; and
	\item For each $C_i \in \mathscr{C}$, there exists $C_j \in \mathscr{C}$ such that $(C_i, C_j) \in    \mathscr{R}$ or $(C_j, C_i) \in \mathscr{R}$.
\end{enumerate}

Given a set $I$, $\mathbb{R}^I$ refers to the vector space of real-valued functions with domain $I$. If $x$ is a vector in $\mathbb{R}^I$, we use the symbol $x_i$ to denote the number that $x$ assigns to $i \in I$. In the context of a CRN, $\mathbb{R}^{\mathscr{S}}$, $\mathbb{R}^{\mathscr{C}}$, and $\mathbb{R}^{\mathscr{R}}$ are referred to as the \textbf{species space}, \textbf{complex space}, and \textbf{reaction space}, respectively.

In a CRN, we denote the species as $X_1, \ldots, X_m$. This way, $X_i$ can be identified with the vector in $\mathbb{R}^m$ with 1 in the $i$th coordinate and zero elsewhere. We denote the reactions as $R_1, \ldots, R_r$. We denote the complexes as $C_1, \ldots, C_n$ where the manner in which the complexes are numbered play no essential role. A complex $C_i \in \mathscr{C}$ is given as $\displaystyle C_i = \sum_{j=1}^m c_{ij} X_j$ or as the vector $(c_{i1}, \ldots, c_{im}) \in \mathbb{R}^m$. The coefficient $c_{ij}$ is called the \textbf{stoichiometric coefficient} of species $X_j$ in complex $C_i$. Stoichiometric coefficients are all nonnegative numbers. We define the \textbf{zero complex} as the zero vector in $\mathbb{R}^m$. A reaction $0 \rightarrow X$ is called an \textbf{inflow reaction} while a reaction $X \rightarrow 0$ is called an \textbf{outflow reaction}. The ordered pair $(C_i, C_j)$ corresponds to the familiar notation $C_i \rightarrow C_j$ which indicates the reaction where complex $C_i$ reacts to complex $C_j$. We call $C_i$ the \textbf{reactant complex} and $C_j$ the \textbf{product complex}. We denote the number of reactant complexes as $n_r$. A reaction $C_i \rightarrow C_j$ is called \textbf{reversible} if it is accompanied by its reverse reaction $C_j \rightarrow C_i$. Otherwise, it is called \textbf{irreversible}.

\bigskip

\begin{example}
	Consider the reaction
	\begin{equation*}
		2 X_1 + X_2 \rightarrow 2 X_3.
	\end{equation*}
	$X_1$, $X_2$, and $X_3$ are the species, the reactant complex is $2 X_1 + X_2$, and $2 X_3$ is the product complex. The stoichiometric coefficients are 2, 1, and 2 for $X_1$, $X_2$, and $X_3$, respectively.
\end{example}

Let $\mathscr{N} = (\mathscr{S}, \mathscr{C}, \mathscr{R})$ be a CRN. For each reaction $C_i \rightarrow C_j \in \mathscr{R}$, we associate the \textbf{reaction vector} $C_j - C_i \in \mathbb{R}^m$. The linear subspace of $\mathbb{R}^m$ spanned by the reaction vectors is called the \textbf{stoichiometric subspace} of $\mathscr{N}$, defined as $S = \text{span}\{C_j - C_i \in \mathbb{R}^m \mid C_i \rightarrow C_j \in \mathscr{R}\}$. The \textbf{rank} of $\mathscr{N}$ is given by $s = \text{dim}(S)$, i.e., the rank of the network is the rank of its set of reaction vectors. In this paper, we sometimes use the notation $\mathscr{N} = \{R_1, \ldots, R_r\}$ where we loosely use the notation $R_i$ to refer to either reaction $i$ or its corresponding reaction vectors.

Two vectors $x^*, x^{**} \in \mathbb{R}^m$ are said to be \textbf{stoichiometrically compatible} if $x^* - x^{**}$ is an element of the stoichiometric subspace $S$. Stoichiometric compatibility is an equivalence relation that induces a partition of $\mathbb{R}_{\geq 0}^{\mathscr{S}}$ or $\mathbb{R}_{>0}^{\mathscr{S}}$ into equivalence classes called the \textbf{stoichiometric compatibility classes} or \textbf{positive stoichiometric compatibility classes}, respectively, of the network. In particular, the stoichiometric compatibility class containing $x \in \mathbb{R}_{\geq 0}^{\mathscr{S}}$ is the set $(x + S) \cap \mathbb{R}_{\geq 0}^{\mathscr{S}}$ where $x + S$ is the left coset of $S$ containing $x$. Similarly, the positive stoichiometric compatibility class containing $x \in \mathbb{R}_{>0}^{\mathscr{S}}$ is the set $(x + S) \cap \mathbb{R}_{>0}^{\mathscr{S}}$.

The \textbf{molecularity matrix} $Y$ is an $m \times n$ matrix whose entry $Y_{ij}$ is the stoichiometric coefficient of species $X_i$ in complex $C_j$. The \textbf{incidence matrix} $I_a$ is an $n \times r$ matrix whose entry $(I_a)_{ij}$ is defined as follows:
\begin{equation*}
	(I_a)_{ij} =
	\left\{
        	\begin{array}{rl}
        		-1	& \text{if $C_i$ is the reactant complex of reaction $R_j$} \\
            	 1	& \text{if $C_i$ is the product complex of reaction $R_j$} \\
        		 0	& \text{otherwise} \\
        	\end{array}
        \right..
\end{equation*}
The \textbf{stoichiometric matrix} $N$ is the $m \times r$ matrix given by $N = Y I_a$. The columns of $N$ are the reaction vectors of the system. From the definition of stoichiometric subspace, we can see that $S$ is the image of $N$, written as $S = \text{Im}(N)$. Observe that $s = \text{dim}(S) = \text{dim}(\text{Im}(N)) = \text{rank}(N)$.

\bigskip

\begin{example}
	\label{ex:2}
	Consider the following CRN:
	\begin{align*}
		& R_1: 2 X_1 \rightarrow X_3 \\
		& R_2: X_2 + X_3 \rightarrow X_3 \\
		& R_3: X_3 \rightarrow X_2 + X_3 \\
		& R_4: 3 X_4 \rightarrow X_2 + X_3 \\
		& R_5: 2 X_1 \rightarrow 3 X_4.
	\end{align*}
	The set of species and complexes are $\mathscr{S} = \{ X_1, X_2, X_3, X_4 \}$ and $\mathscr{C} = \{ 2 X_1, X_2 + X_3, X_3, 3 X_4 \}$, respectively. Thus, there are $m = 4$ species, $n = 4$ complexes, $n_r = 4$ reactant complexes, and $r = 5$ reactions. The network's molecularity matrix, incidence matrix, and stoichiometric matrix are as follows:
	\begin{multicols}{2}
		\begin{center}
			\[ Y =
				\begin{blockarray}{ccccc}
						2 X_1 & X_2 + X_3 & X_3 & 3 X_4	\\
					\begin{block}{[cccc]l}
						2 & 0 & 0 & 0 \topstrut & X_1	\\
						0 & 1 & 0 & 0 & X_2	\\
						0 & 1 & 1 & 0 & X_3 \\
						0 & 0 & 0 & 3 \botstrut	& X_4	\\
					\end{block}
				\end{blockarray}
			\]
		\end{center}
		\begin{center}
			\[ I_a =
				\begin{blockarray}{rrrrrr}
						R_1 & R_2 & R_3 & R_4 & R_5\textcolor{white}{1}	\\
					\begin{block}{[rrrrr]l}
						\textcolor{white}{1}-1 & 0 & 0 & 0 & -1\textcolor{white}{1} \topstrut & 2 X_1	\\
						0 & -1 & 1 & 1 & 0\textcolor{white}{1} & X_2 + X_3	\\
						1 & 1 & -1 & 0 & 0\textcolor{white}{1} & X_3 \\
						0 & 0 & 0 & -1 & 1\textcolor{white}{1} \botstrut	& 3 X_4	\\
					\end{block}
				\end{blockarray}
			\]
		\end{center}
	\end{multicols}
	\begin{equation*}
		N = Y I_a =
		\left[
			\begin{array}{rrrrr}
			-2 & 0 & 0 & 0 & -2 \\
			0 & -1 & \textcolor{white}{-}1 & 1 & 0 \\
			1 & 0 & 0 & 1 & 0 \\
			0 & 0 & 0 & -3 & 3 \\
			\end{array}
		\right].
	\end{equation*}
	The network has rank $s = \text{rank}(N) = 3$.
\end{example}

CRNs can be viewed as directed graphs where the complexes are represented by vertices and the reactions by edges. The \textbf{linkage classes} of a CRN are the subnetworks of its reaction graph where for any complexes $C_i$ and $C_j$ of the subnetwork, there is a path between them. The number of linkage classes is denoted by $\ell$. The linkage class is said to be a \textbf{strong linkage class} if there is a directed path from $C_i$ to $C_j$, and vice versa, for any complexes $C_i$ and $C_j$ of the subnetwork. The number of strong linkage classes is denoted by $s \ell$. Moreover, \textbf{terminal strong linkage classes}, the number of which is denoted as $t$, are the maximal strongly connected subnetworks where there are no edges (reactions) from a complex in the subgraph to a complex outside the subnetwork. Complexes belonging to terminal strong linkage classes are called \textbf{terminal}; otherwise, they are called \textbf{nonterminal}.

In Example \ref{ex:2}, the number of linkage classes is $\ell = 1$: $\{2 X_1, X_3, X_2 + X_3, 3 X_4\}$; the number of strong linkage classes is $s \ell = 3$: $\{X_3, X_2 + X_3\}, \{2 X_1\}, \{3 X_4\}$; and the number of terminal strong linkage classes is $t = 1$: $\{X_3, X_2 + X_3\}$. $X_3$ and $X_2 + X_3$ are terminal complexes while $2 X_1$ and $3 X_4$ are nonterminal complexes.

A CRN is called \textbf{weakly reversible} if $s \ell = \ell$, \textbf{$t$-minimal} if $t = \ell$, \textbf{point terminal} if $t = n - n_r$, and \textbf{cycle terminal} if $n - n_r = 0$. The \textbf{deficiency} of a CRN is given by $\delta = n - \ell - s$.

For a CRN $\mathscr{N}$, the linear subspace of $\mathbb{R}^m$ generated by the reactant complexes is called the \textbf{reactant subspace} of $\mathscr{N}$, defined as $R = \text{span}\{C_i \in \mathbb{R}^m \mid C_i \rightarrow C_j \in \mathscr{R}\}$. The \textbf{reactant rank} of $\mathscr{N}$ is given by $q = \text{dim}(R)$, i.e., the reactant rank of the network is the rank of its set of complexes. The \textbf{reactant deficiency} of $\mathscr{N}$ is given by $\delta_p = n_r - q$.

To make sense of the reactant subspace $R$, write the incidence matrix as $I_a = I_a^+ - I_a^-$ where $I_a^+$ consists only of the 0's and 1's in $I_a$ while $I_a^-$ contains only the 0's and absolute values of the $-1$'s. We form the \textbf{reactant matrix} $N^-$ (size $m \times r$) given by $N^- = Y I_a^-$. The columns of $N^-$ contains the reactant complexes of the system. From the definition of reactant subspace, we can see that $R$ is the image of $N^-$, written as $R = \text{Im}(N^-)$. Observe that $q = \text{dim}(R) = \text{dim}(\text{Im}(N^-)) = \text{rank}(N^-)$.

The incidence matrix of the network in Example \ref{ex:2} can be written as
\begin{align*}
	I_a &= I_a^+ - I_a^- \\
	\left[
		\begin{array}{rrrrr}
			-1 & 0 & 0 & 0 & -1 \\
			0 & -1 & 1 & 1 & 0 \\
			1 & 1 & -1 & 0 & 0 \\
			0 & 0 & 0 & -1 & 1 \\
		\end{array}
	\right]
	&=
	\left[
		\begin{array}{rrrrr}
			0 & 0 & 0 & 0 & 0 \\
			0 & 0 & 1 & 1 & 0 \\
			1 & 1 & 0 & 0 & 0 \\
			0 & 0 & 0 & 0 & 1 \\
		\end{array}
	\right]
	-
	\left[
		\begin{array}{rrrrr}
			1 & 0 & 0 & 0 & 1 \\
			0 & 1 & 0 & 0 & 0 \\
			0 & 0 & 1 & 0 & 0 \\
			0 & 0 & 0 & 1 & 0 \\
		\end{array}
	\right]
\end{align*}
allowing us to form the reactant matrix
\begin{equation*}
	N^- = Y I_a^- =
	\left[
		\begin{array}{rrrrr}
		2 & 0 & 0 & 0 & 2 \\
		0 & 1 & 0 & 0 & 0 \\
		0 & 1 & 1 & 0 & 0 \\
		0 & 0 & 0 & 3 & 0 \\
		\end{array}
	\right].
\end{equation*}
The network has reactant rank $q = \text{rank}(N^-) = 4$ and reactant deficiency $\delta_p = n_r - q = 4 - 4 = 0$.

\subsection{Chemical Kinetic Systems}

A \textbf{kinetics} $K$ for a CRN $\mathscr{N} = (\mathscr{S}, \mathscr{C}, \mathscr{R})$ is an assignment to each reaction $C_i \rightarrow C_j \in \mathscr{R}$ of a rate function $K_{C_i \rightarrow C_j}: \mathbb{R}_{\geq 0}^{\mathscr{S}} \rightarrow \mathbb{R}_{\geq 0}$ such that
\begin{equation*}
	K_{C_i \rightarrow C_j} (X) > 0 \text{ if and only if } \text{supp}(C_i) \subset \text{supp}(X).
\end{equation*}
The system $(\mathscr{N}, K)$ is called a \textbf{chemical kinetic system} (CKS). The \textbf{support} of complex $C_i \in \mathscr{C}$ is $\text{supp}(C_i) = \{X_j \in \mathscr{S} \mid c_{ij} \neq 0\}$, i.e, it is the set of all species that have nonzero stoichiometric coefficients in complex $C_i$.

A kinetics gives rise to two closely related objects: the species formation rate function and the associated ordinary differential equation system.

The \textbf{species formation rate function} (SFRF) of a CKS is given by
\begin{equation*}
	f(X) = \sum_{C_i \rightarrow C_j} K_{C_i \rightarrow C_j} (X) (C_j - C_i)
\end{equation*}
where $X$ is the vector of species in $\mathscr{S}$ and $K_{C_i \rightarrow C_j}$ is the rate function assigned to reaction $C_i \rightarrow C_j \in \mathscr{R}$. The SFRF is simply the summation of the reaction vectors for the network, each multiplied by the corresponding rate function. The \textbf{kinetic subspace} $\mathcal{K}$ for a CKS is the linear subspace of $\mathbb{R}^{\mathscr{S}}$ defined by $\mathcal{K} = \text{span}\{ \text{Im}(f) \}.$ Note that the SFRF can be written as $f(X) = N K(X)$ where $K$ the vector of rate functions. The equation $\dot{X} = f(X)$ is the \textbf{ordinary differential equation} (ODE) \textbf{system} or \textbf{dynamical system} of the CKS.

The ODE system of the CRN in Example \ref{ex:2} can be written as
\begin{equation*}
	\dot{X} =
	\left[
		\begin{array}{c}
			\dot{X}_1 \\
			\dot{X}_2 \\
			\dot{X}_3 \\
			\dot{X}_4 \\
		\end{array}
	\right]
	=
	\left[
		\begin{array}{rrrrr}
		-2 & 0 & 0 & 0 & -2 \\
		0 & -1 & \textcolor{white}{-}1 & 1 & 0 \\
		1 & 0 & 0 & 1 & 0 \\
		0 & 0 & 0 & -3 & 3 \\
		\end{array}
	\right]
	\left[
		\begin{array}{l}
			k_1 X_1^{f_{11}} \\
			k_2 X_2^{f_{22}} X_3^{f_{23}} \\
			k_3 X_3^{f_{33}} \\
			k_4 X_4^{f_{44}} \\
			k_5 X_1^{f_{51}} \\
		\end{array}
	\right]
	= N K(X).
\end{equation*}

A zero of the SFRF is called an \textbf{equilibrium} or a \textbf{steady state} of the system. If $f$ is differentiable, an equilibrium $X^*$ is called \textbf{degenerate} if $\text{Ker}(J_{X^*} (f)) \cap S \neq \{0\}$ where $J_{X^*} (f)$ is the Jacobian of $f$ evaluated at $X^*$ and Ker is the kernel function; otherwise, the equilibrium is said to be \textbf{nondegenerate}.

A vector $X \in \mathbb{R}_{>0}^{\mathscr{S}}$ is called \textbf{complex balanced} if $K(X)$ is contained in $\text{Ker}(I_a)$ where $I_a$ is the incidence matrix. Furthermore, if $X$ is a positive equilibrium, then we call it a \textbf{complex balanced equilibrium}. A CKS is called \textbf{complex balanced} if it has a complex balanced equilibrium.

The reaction vectors of a CRN are \textbf{positively dependent} if, for each reaction $C_i \rightarrow C_j \in \mathscr{R}$, there exists a positive number $\alpha_{C_i \rightarrow C_j}$ such that
\begin{equation*}
	\sum_{C_i \rightarrow C_j} \alpha_{C_i \rightarrow C_j} (C_j - C_i) = 0.
\end{equation*}
A CRN with positively dependent reaction vectors is said to be \textbf{positive dependent}. Shinar and Feinberg \cite{SHFE2012} showed that a CKS can admit a positive equilibrium only if its reaction vectors are positively dependent. The \textbf{set of positive equilibria} of a CKS is given by
\begin{equation*}
	E_+ (\mathscr{N}, K) = \{X \in \mathbb{R}_{>0}^{\mathscr{S}} \mid f(X) = 0\}.
\end{equation*}
A CRN is said to \textbf{admit multiple (positive) equilibria} if there exist positive rate constants such that the ODE system admits more than one stoichiometrically compatible equilibria. Analogously, the \textbf{set of complex balanced equilibria} of a CKS $(\mathscr{N}, K)$ is given by
\begin{equation*}
	Z_+ (\mathscr{N}, K) = \{ X \in \mathbb{R}_{>0}^{\mathscr{S}} \mid I_a K(X) = 0 \} \subseteq E_+ (\mathscr{N}, K).
\end{equation*}

Let $F$ be an $r \times m$ matrix of real numbers. Define $X^F$ by $\displaystyle (X^F)_i = \prod_{j=1}^m X_j^{f_{ij}}$ for $i = 1, \ldots, r$. A \textbf{power law kinetics} (PLK) assigns to each $i$th reaction a function
\begin{equation*}
	K_i (X) = k_i (X^F)_i
\end{equation*}
with \textbf{rate constant} $k_i > 0$ and \textbf{kinetic order} $f_{ij} \in \mathbb{R}$. The vector $k \in \mathbb{R}^r$ is called the \textbf{rate vector} and the matrix $F$ is called the \textbf{kinetic order matrix}. We refer to a CRN with PLK as a \textbf{power law system}. The PLK becomes the well-known \textbf{mass action kinetics} (MAK) if the kinetic order matrix consists of stoichiometric coefficients of the reactants. We refer to a CRN with MAK as a \textbf{mass action system}.

In the ODE system of Example \ref{ex:2}, we assumed PLK so that the kinetic order matrix is
\begin{equation*}
	F =
	\left[
		\begin{array}{cccc}
			f_{11} & 0 & 0 & 0 \\
			0 & f_{22} & f_{23} & 0 \\
			0 & 0 & f_{33} & 0 \\
			0 & 0 & 0 & f_{44} \\
			f_{51} & 0 & 0 & 0 \\
		\end{array}
	\right]
\end{equation*}
where $f_{ij} \in \mathbb{R}$. If we assume MAK, the kinetic order matrix is
\begin{equation*}
	F =
	\left[
		\begin{array}{cccc}
			2 & 0 & 0 & 0 \\
			0 & 1 & 1 & 0 \\
			0 & 0 & 1 & 0 \\
			0 & 0 & 0 & 3 \\
			2 & 0 & 0 & 0 \\
		\end{array}
	\right].
\end{equation*}

The reactions $R_i, R_j \in \mathscr{R}$ are called \textbf{branching reactions} if they have the same reactant complex. One way to check if we have identified all branching reactions is through the formula
\begin{equation*}
    r - n_r = \sum_{C_i} (\vert R_{C_i} \vert - 1)
\end{equation*}
where $C_i$ is the reactant complex in $C_i \rightarrow C_j \in \mathscr{R}$, $R_{C_i}$ is the set of branching reactions of $C_i$, and $\vert R_{C_i} \vert$ is the cardinality of $R_{C_i}$. $r - n_r = 0$ if and only if all reactant complexes are nonbranching. A CRN is called \textbf{branching} if $r > n_r$.

We can classify a power law system based on the kinetic orders assigned to its branching reactions. A power law system has \textbf{reactant-determined kinetics} (of type PL-RDK) if, for any two branching reactions $R_i, R_j \in \mathscr{R}$, their corresponding rows of kinetic orders in $F$ are identical, i.e., $f_{ik} = f_{jk}$ for $k = 1, \ldots, m$. Otherwise, a power law system has \textbf{non-reactant-determined kinetics} (of type PL-NDK). From Proposition 12 in Arceo et al \cite{AJLM2017b}, a nonbranching power law system is PL-RDK.

Note that the stoichiometric subspace $S$ is just the set of all linear combinations of the reaction vectors, i.e., the set of all vectors in $\mathbb{R}^{\mathscr{S}}$ can be written in the form
\begin{equation*}
	\sum_{C_i \rightarrow C_j} \alpha_{C_i \rightarrow C_j} (C_j - C_i).
\end{equation*}

Let $L: \mathbb{R}^{\mathscr{R}} \rightarrow S$ be the linear map defined by
\begin{equation*}
	L(\alpha) = \sum_{C_i \rightarrow C_j} \alpha_{C_i \rightarrow C_j} (C_j - C_i).
\end{equation*}
$\text{Ker}(L)$ is the set of all vectors $\alpha \in \mathbb{R}^{\mathscr{R}}$ such that $L(\alpha) = 0$.

We say that a CRN is \textbf{concordant} if there do not exist an $\alpha \in \text{Ker}(L)$ and a nonzero $\sigma \in S$ having the following properties:
\begin{enumerate}[($i$)]
	\item For each $C_i \rightarrow C_j \in \mathscr{R}$ such that $\alpha_{C_i \rightarrow C_j} \neq 0$, $\text{supp}(C_i)$ contains a species $X$ for which $\text{sgn}(\sigma_X) = \text{sgn}(\alpha_{C_i \rightarrow C_j})$ where $\sigma_X$ denotes the term in $\sigma$ involving the species $X$ and $\text{sgn}(\cdot)$ is the signum function.
	\item For each $C_i \rightarrow C_j \in \mathscr{R}$ such that $\alpha_{C_i \rightarrow C_j} = 0$, either $\sigma_X = 0$ for all $X \in \text{supp}(C_i)$, or else $\text{supp}(C_i)$ contains species $X$ and $X'$ for which $\text{sgn}(\sigma_X) = -\text{sgn}(\sigma_{X'})$, but not zero.
\end{enumerate}
A network that is not concordant is \textbf{discordant}.

A CKS is \textbf{injective} if, for each pair of distinct stoichiometrically compatible vectors \linebreak $X^*, X^{**} \in \mathbb{R}_{\geq 0}^{\mathscr{S}}$, at least one of which is positive,
\begin{equation*}
	\sum_{C_i \rightarrow C_j} K_{C_i \rightarrow C_j} (X^{**}) (C_j - C_i) \neq \sum_{C_i \rightarrow C_j} K_{C_i \rightarrow C_j} (X^*) (C_j - C_i).
\end{equation*}
Clearly, an injective kinetic system cannot admit two distinct stoichiometrically compatible equilibria, at least one of which is positive.

A kinetics for a CRN is \textbf{weakly monotonic} if, for each pair of vectors $X^*, X^{**} \in \mathbb{R}_{\geq 0}^{\mathscr{S}}$, the following implications hold for each reaction $C_i \rightarrow C_j \in \mathscr{R}$ such that $\text{supp}(C_i) \subset \text{supp}(X^*)$ and $\text{supp}(C_i) \subset \text{supp}(X^{**})$:
\begin{enumerate}[($i$)]
	\item $K_{C_i \rightarrow C_j} (X^{**}) > K_{C_i \rightarrow C_j} (X^*)$ implies that there is a species $X_k \in \text{supp}(C_i)$ with $X_k^{**} > X_k^*$.
	\item $K_{C_i \rightarrow C_j} (X^{**}) = K_{C_i \rightarrow C_j} (X^*)$ implies that $X_k^{**} = X_k^*$ for all $X_k \in \text{supp}(C_i)$ or else there are species \linebreak $X_k, X_k' \in \text{supp}(C_i)$ with $X_k^{**} > X_k^*$ and $(X_k')^{**} < (X_k')^*$.
\end{enumerate}
We say that a CKS is \textbf{weakly monotonic} when its kinetics is weakly monotonic.

\bigskip

\begin{example}
	Every MAK is weakly monotonic.
\end{example}

\subsection{Decomposition Theory}

A \textbf{covering} of a CRN is a collection of subsets $\{\mathscr{R}_1, \ldots, \mathscr{R}_k\}$ whose union is $\mathscr{R}$. A covering is called a \textbf{decomposition} of $\mathscr{N}$ if the sets $\mathscr{R}_i$ form a partition of $\mathscr{R}$. $\mathscr{R}_i$ defines a subnetwork $\mathscr{N}_i$ of $\mathscr{N}$ where $\mathscr{N}_i = (\mathscr{S}_i, \mathscr{C}_i, \mathscr{R}_i)$ such that $\mathscr{C}_i $ consists of all complexes occurring in $\mathscr{R}_i$ and $\mathscr{S}_i$ has all the species occurring in $\mathscr{C}_i$. In this paper, we will denote a decomposition as a union of the subnetworks: $\mathscr{N} = \mathscr{N}_1 \cup \ldots \cup \mathscr{N}_k$. We refer to a ``decomposition'' with a single ``subnetwork'' $\mathscr{N} = \mathscr{N}_1$ as the \textbf{trivial decomposition}. Furthermore, when a network has been decomposed into subnetworks, we can refer to the said network as the \textbf{parent network}.

The most widely used decomposition of a reaction network is the set of linkage classes. Linkage classes have the special property that they not only partition the set of reactions but also the set of complexes.

A decomposition $\mathscr{N} = \mathscr{N}_1 \cup \ldots \cup \mathscr{N}_k$ is \textbf{independent} if the parent network's stoichiometric subspace $S$ is the direct sum of the subnetworks' stoichiometric subspaces $S_i$. Equivalently, the sum is direct when the rank of the parent network is equal to the sum of the ranks of the individual subnetworks, i.e.,
\begin{equation*}
	s = \sum_{i=1}^k s_i \text{ where } s_i = \text{dim}(S_i).
\end{equation*}

A network decomposition $\mathscr{N} = \mathscr{N}_1 \cup \ldots \cup \mathscr{N}_k$ is a \textbf{refinement} of $\mathscr{N} = \mathscr{N}_1' \cup \ldots \cup \mathscr{N}_{k'}'$ (and the latter a \textbf{coarsening} of the former) if it is induced by a refinement $\{\mathscr{R}_1, \ldots, \mathscr{R}_k\}$ of $\{\mathscr{R}_1', \ldots, \mathscr{R}_{k'}'\}$.

\bigskip

\begin{example}
	If $\mathscr{N} = \mathscr{N}_1 \cup \mathscr{N}_2 \cup \mathscr{N}_3$ and $\mathscr{N}' = \mathscr{N}_2 \cup \mathscr{N}_3$, then
	\begin{itemize}
		\item $\mathscr{N} = \mathscr{N}_1 \cup \mathscr{N}_2 \cup \mathscr{N}_3$ is a refinement of $\mathscr{N} = \mathscr{N}_1 \cup \mathscr{N}'$; and
		\item $\mathscr{N} = \mathscr{N}_1 \cup \mathscr{N}'$ is a coarsening of $\mathscr{N} = \mathscr{N}_1 \cup \mathscr{N}_2 \cup \mathscr{N}_3$.
	\end{itemize}
\end{example}

A decomposition is said to be \textbf{incidence independent} if its incidence matrix is the direct sum of the incidence matrices of the subnetworks. Consider the network decomposition \linebreak $\mathscr{N} = \mathscr{N}_1 \cup \ldots \cup \mathscr{N}_k$. If $\mathscr{N}$ ($\mathscr{N}_i$) has $n$ ($n_i$) complexes and $\ell$ ($\ell_i$) linkage classes, then incidence independence is equivalent to
\begin{equation*}
    n - \ell = \sum_{i=1}^k (n_i - \ell_i).
\end{equation*}

Despite its early founding by Feinberg in 1987, decomposition theory has received little attention in the Chemical Reaction Network Theory (CRNT) community. Recently, however, its usefulness expanded beyond results on equilibria existence and parametrization, e.g., in the relatively new field of concentration robustness. Interesting results for large and high deficiency systems have been derived for more general kinetic systems such as power law systems \cite{FOMF2021, FOME2021} and Hill-type systems \cite{HEME2021}.

\section{Variable Descriptions and Units}
\label{app:var}

The following are the variables used in the metabolic insulin signaling network together with their descriptions and unit of measurement:
\begin{align*}
	& X_2 = \text{Unbound surface insulin receptors (in molar)} \\
	& X_3 = \text{Unphosphorylated once-bound surface receptors (in molar)} \\
	& X_4 = \text{Phosphorylated twice-bound surface receptors (in molar)} \\
	& X_5 = \text{Phosphorylated once-bound surface receptors (in molar)} \\
	& X_6 = \text{Unbound unphosphorylated intracellular receptors (in molar)} \\
	& X_7 = \text{Phosphorylated twice-bound intracellular receptors (in molar)} \\
	& X_8 = \text{Phosphorylated once-bound intracellular receptors (in molar)} \\
	& X_9 = \text{Unphosphorylated IRS-1 (in molar)} \\
	& X_{10} = \text{Tyrosine-phosphorylated IRS-1 (in molar)} \\
	& X_{11} = \text{Unactivated PI 3-kinase (in molar)} \\
	& X_{12} = \text{Tyrosine-phosphorylated IRS-1/activated PI 3-kinase complex (in molar)} \\
	& X_{13} = \text{PI(3,4,5)P$_3$ out of the total lipid population (in \%)} \\
	& X_{14} = \text{PI(4,5)P$_2$ out of the total lipid population (in \%)} \\
	& X_{15} = \text{PI(3,4)P$_2$ out of the total lipid population (in \%)} \\
	& X_{16} = \text{Unactivated Akt (in \%)} \\
	& X_{17} = \text{Activated Akt (in \%)} \\
	& X_{18} = \text{Unactivated PKC-$\zeta$ (in \%)} \\
	& X_{19} = \text{Activated PKC-$\zeta$ (in \%)} \\
	& X_{20} = \text{Intracellular GLUT4 (in \%)} \\
	& X_{21} = \text{Cell surface GLUT4 (in \%)}.
\end{align*}

\section{Determining the CRN of the Metabolic Insulin Signaling Network}
\label{app:crn}

In Section 2.2, it is determined by the Hars-T\'{o}th criterion that there is a reaction network of the form $AX \xrightarrow{k} BX$ such that $f(X) = (B - A)^\top (k \circ X^A)$ where $A$ and $B$ have nonnegative integer entries. In this section, we show the detailed computation on how to get matrices $A$ and $B$.

To determine matrix $A$, we write
\begin{equation*}
	k \circ X^A =
	\left[
	\begin{array}{l}
		k_1 X_2 \\
		k_2 X_3 \\
		k_3 X_5 \\
		\textcolor{white}{k_nl} \vdots \\
		k_{34} \\
		k_{35} X_{20}
	\end{array}
	\right]
\end{equation*}
which gives us
\begin{center}
	\[ A =
		\begin{blockarray}{ccccccl}
				X_2 & X_3 & X_4 & \ldots & X_{20} & X_{21} \\
			\begin{block}{[cccccc]l}
				1 & 0 & 0 & \ldots & 0 & 0 \topstrut & R_1	\\
				0 & 1 & 0 & \ldots & 0 & 0 & R_2	\\
				0 & 0 & 0 & \ldots & 0 & 0 & R_3	\\
				\vdots & \vdots & \vdots & \cdots & \vdots & \vdots & \textcolor{white}{1}\vdots \\
				0 & 0 & 0 & \ldots & 0 & 0 & R_{34}	\\
				0 & 0 & 0 & \ldots & 1 & 0 \botstrut & R_{35}	\\
			\end{block}
		\end{blockarray}
	\]
\end{center}

To get $B$, observe first that we can write the ODE system as
\begin{equation*}
	\dot{X} = f(X) =
	\left[
	\begin{array}{rrrrrr}
		-1 & 1 & 0 & \ldots & 0 & 0	\\
		1 & -1 & 0 & \ldots & 0 & 0	\\
		0 & 0 & \textcolor{white}{-}1 & \ldots & \textcolor{white}{-}0 & 0	\\
		\vdots & \vdots & \vdots & \cdots & \vdots & \vdots \\
		0 & 0 & 0 & \ldots & 1 & -1 \\
		0 & 0 & 0 & \ldots & 0 & 0
	\end{array}
	\right]
	\left[
	\begin{array}{l}
		k_1 X_2 \\
		k_2 X_3 \\
		k_3 X_5 \\
		\textcolor{white}{k_nl} \vdots \\
		k_{34} \\
		k_{35} X_{20}
	\end{array}
	\right]
\end{equation*}
providing us with
\begin{equation*}
	(B - A)^T =
	\left[
	\begin{array}{rrrrrr}
		-1 & 1 & 0 & \ldots & 0 & 0	\\
		1 & -1 & 0 & \ldots & 0 & 0	\\
		0 & 0 & \textcolor{white}{-}1 & \ldots & \textcolor{white}{-}0 & 0	\\
		\vdots & \vdots & \vdots & \cdots & \vdots & \vdots \\
		0 & 0 & 0 & \ldots & 1 & -1 \\
		0 & 0 & 0 & \ldots & 0 & 0
	\end{array}
	\right]
\end{equation*}
allowing us to easily get $B$:
\begin{center}
	\[ B =
		\begin{blockarray}{ccccccl}
				X_2 & X_3 & X_4 & \ldots & X_{20} & X_{21} \\
			\begin{block}{[cccccc]l}
				0 & 1 & 0 & \ldots & 0 & 0 \topstrut & R_1	\\
				1 & 0 & 0 & \ldots & 0 & 0 & R_2	\\
				0 & 0 & 1 & \ldots & 0 & 0 & R_3	\\
				\vdots & \vdots & \vdots & \cdots & \vdots & \vdots & \textcolor{white}{1}\vdots \\
				0 & 0 & 0 & \ldots & 1 & 0 & R_{34}	\\
				0 & 0 & 0 & \ldots & 0 & 0 \botstrut & R_{35}	\\
			\end{block}
		\end{blockarray}.
	\]
\end{center}

Matrices $A$ and $B$ are used to determine the CRN corresponding to the ODE system \linebreak $\dot{X} = f(X)$ in Section 2.1.

\section{Generalization of Lemma 3}
\label{app:lem}

In this section, we present two propositions related to Lemma 3 in Section 3.2. The first can be viewed as a corollary of statement $(i)$ of Lemma 3. The rest of the section leads to a generalization of this corollary to the case of a nonempty intersection of the species sets.

\begin{proposition}
	Let $\mathscr{N} = (\mathscr{S}, \mathscr{C}, \mathscr{R})$ be a CRN with kinetics $K$ and subsets $\mathscr{S}_1$ and $\mathscr{S}_2$ such that $\mathscr{S}_1 \cup \mathscr{S}_2 = \mathscr{S}$. If $\mathscr{S}_1 \cap \mathscr{S}_2 = \varnothing$, then for any kinetics, if each subsystem has an equilibrium in every stoichiometric class, then so does $(\mathscr{N}, K)$. If the equilibria are unique positive equilibria, the same holds for $(\mathscr{N}, K)$.
\end{proposition}

\smallskip

\begin{proof}
	By Feinberg's Decomposition Theorem, any equilibrium of the whole system is also an equilibrium in each subnetwork. Hence, such an equilibrium in a stoichiometric compatibility class corresponds to its projection pair in the subnetwork stoichiometric compatibility classes.
\end{proof}

The next proposition describes the relationship between the stoichiometric classes of a CRN and those of a binary decomposition. This lays the foundation for the generalization of statement $(i)$ of Lemma 3 to the nonempty intersection case.

\begin{proposition}
	Let $\mathscr{N} = \mathscr{N}_1 \cup \mathscr{N}_2$ be a decomposition and $S$, $S_1$, and $S_2$ the corresponding stoichiometric subspaces. Let $d: \mathbb{R}^{\mathscr{S}} \rightarrow \mathbb{R}^{\mathscr{S}} / S_1 \times \mathbb{R}^{\mathscr{S}} / S_2$ be the linear map $d(x) = (x + S_1, x + S_2)$ and $\Delta: \mathbb{R}^{\mathscr{S}} / S_1 \times \mathbb{R}^{\mathscr{S}} / S_2 \rightarrow \mathbb{R}^{\mathscr{S}} / S $ the linear map $\Delta (x_1 + S_1, x_2 + S_2) = (x_1 - x_2) + S$. Then
	\begin{enumerate}[$(i)$]
		\item $\Delta$ is surjective.
		\item $Ker(\Delta) = \{ (x_1 + S_1, x_2 + S_2) \mid (x_1 + S_1) \cap (x_2 + S_2) \neq \varnothing \}$
		\item $Im(d) = Ker(\Delta)$
		\setcounter{nameOfYourChoice}{\value{enumi}}
	\end{enumerate}
	If, in addition, the decomposition is independent, then
	\begin{enumerate}[$(i)$]
		\setcounter{enumi}{\value{nameOfYourChoice}}
		\item $d: \mathbb{R}^{\mathscr{S}} \rightarrow Ker(\Delta)$ is an isomorphism.
		\item $\vert (x_1 + S_1) \cap (x_2 + S_2) \vert \leq 1$
	\end{enumerate}
\end{proposition}

\smallskip

\begin{proof}
	It is easily verified that the maps are well-defined.

	$(i)$ For any $z + S \in \mathbb{R}^{\mathscr{S}} / S$, write $z = z' + z''$ where $z' \in S$ and $z'' \in S^\perp$. Let \linebreak $z' = z'_1 + z'_2$ where $z'_i = p_i (z')$ (the projection map $p_i$ is defined in statement $(i)$ of Lemma 3). Then $\Delta (\frac{z''}{2} + z'_2 + S_1, -\frac{z''}{2} - z'_1 + S_2) = z'' + z'_2 + z'_1 + S$.
	
	$(ii)$  By definition, we have $\text{Ker}(\Delta) = \{ (x_1 + S_1, x_2 + S_2) \mid x_1 - x_2 \in S \}$. This implies that $x_1 - x_2 = s_1 + s_2$ where $s_i \in S_i$. Therefore, $x = x_1 - s_1 = x_2 + s_2 \in (x_1 + S_1) \cap (x_2 + S_2)$. Conversely, $x = x_1 + s_1 = x_2 + s_2$ implies that $x_1 - x_2 = (-s_1) + s_2 \in S$.
	
	$(iii)$  Clearly, $\text{Im}(d) \subseteq \text{Ker}(\Delta)$. On the other hand, for any $(x_1 + S_1, x_2 + S_2) \in \text{Ker}(\Delta)$ and any \linebreak $x \in (x_1 + S_1) \cap (x_2 + S_2)$, $x + S_1 = x_1 + S_1$ and $x + S_2 = x_2 + S_2$. Hence, $(x_1 + S_1, x_2 + S_2) = d(x)$.
	
	$(iv)$  $d(x) = (S_1, S_2) \Leftrightarrow x \in S_1$ and $x \in S_2 \Rightarrow x = 0$ since the sum is direct.
	
	$(v)$  Suppose the intersection is nonempty with common elements $x$ and $x'$. Then $x - x' \in S_1$ and $x - x' \in S_2$ implying that $x - x' = 0$ since the sum is direct.
\end{proof}

Hence, for any binary independent decomposition, we have
\begin{equation*}
	\mathbb{R}^{\mathscr{S}} / S_1 \times \mathbb{R}^{\mathscr{S}} / S_2 = \text{Ker}(\Delta) \times \mathbb{R}^{\mathscr{S}} / S.
\end{equation*}
Since $\mathbb{R}^{\mathscr{S}} / S_i = \mathbb{R}^{\mathscr{S}} / S_i \times \mathbb{R}^{m - m_i}$ where $m$ and $m_i$ are the numbers of species in $\mathscr{S}$ and $\mathscr{S}_i$, respectively, we have
\begin{align*}
	\mathbb{R}^{\mathscr{S}} / S_1 \times \mathbb{R}^{\mathscr{S}} / S_2
	& = \mathbb{R}^{\mathscr{S}} / S_1 \times \mathbb{R}^{\mathscr{S}} / S_2 \times \mathbb{R}^{m - m_1} \times \mathbb{R}^{m - m_2} \\
	& = \mathbb{R}^{\mathscr{S}} / S_1 \times \mathbb{R}^{\mathscr{S}} / S_2 \times \mathbb{R}^{m - c}
\end{align*}
where $2m - m_1 - m_2 = m - c$ ($c$ is the number of common species). Since $d$ is an isomorphism between $\mathbb{R}^m$ and $\text{Ker}(\Delta)$, we have
\begin{equation*}
	\mathbb{R}^{\mathscr{S}} / S_1 \times \mathbb{R}^{\mathscr{S}} / S_2 = \mathbb{R}^c \times \mathbb{R}^{\mathscr{S}} / S.
\end{equation*}
This reduces to statement $(i)$ of Lemma 3 when $c = 0$.

\end{appendices}

\end{document}